\documentclass[journal,12pt,onecolumn,draftclsnofoot]{IEEEtran}

\usepackage{cite}
\usepackage{graphicx}
\usepackage{algorithm,algorithmic}
\usepackage[caption=false, font=footnotesize]{subfig}
\usepackage{url}
\usepackage{mathrsfs}
\usepackage{amsfonts}
\usepackage{amsmath}

\DeclareMathOperator*{\argmin}{arg\,min}
\usepackage{amssymb}
\usepackage{mathtools}
\usepackage{algorithm,algorithmic}
\usepackage{amsthm}
\newtheorem{theorem}{Theorem}
\newtheorem{lemma}{Lemma}

\newtheorem{assumption}{Assumption}
\newtheorem{definition}{Definition}
\usepackage{bm}
\renewcommand{\vec}[1]{\boldsymbol{#1}}
\newcommand{\norm}[1]{\left\lVert#1\right\rVert}
\newcommand{\algorithmicbreak}{\textbf{break}}
\newcommand{\BREAK}{\STATE \algorithmicbreak}

\usepackage{etoolbox}
\AtBeginEnvironment{algorithm}{\linespread{0.95}\selectfont}
\usepackage{xcolor}
\usepackage{multirow}



\begin{document}

\title{Joint Device Scheduling and Resource Allocation for Latency Constrained Wireless Federated Learning}

\author{Wenqi Shi,
        Sheng Zhou,~\IEEEmembership{Member,~IEEE},
        Zhisheng Niu,~\IEEEmembership{Fellow,~IEEE}, \linebreak Miao Jiang, and Lu Geng

\thanks{This work is sponsored in part by the National Key R\&D Program of China 2018YFB1800800 and 2018YFB0105005, the Nature Science Foundation of China (No. 61871254, No. 91638204, No. 61861136003), and Hitachi Ltd. Part of this work has been accepted in IEEE ICC 2020 \cite{shi2019device}. (corresponding author: Sheng Zhou)}
\thanks{W. Shi, S. Zhou, and Z. Niu are with the Beijing National
Research Center for Information Science and Technology, Department of Electronic Engineering, Tsinghua University, Beijing 100084,
China (e-mail: swq17@mails.tsinghua.edu.cn; sheng.zhou@tsinghua.edu.cn;
niuzhs@tsinghua.edu.cn).}
\thanks{M. Jiang and L. Geng are with Hitachi (China) Research \& Development Cooperation, Beijing 100190, China (e-mail: miaojiang@hitachi.cn; lgeng@hitachi.cn).}
}


\maketitle

\begin{abstract}
In federated learning (FL), devices contribute to the global training by uploading their local model updates via wireless channels. Due to limited computation and communication resources, device scheduling is crucial to the convergence rate of FL. In this paper, we propose a joint device scheduling and resource allocation policy to maximize the model accuracy within a given \emph{total training time} budget for \emph{latency constrained} wireless FL. A lower bound on the reciprocal of the training performance loss, in terms of the number of training rounds and the number of scheduled devices per round, is derived. Based on the bound, the accuracy maximization problem is solved by decoupling it into two sub-problems. First, given the scheduled devices, the optimal bandwidth allocation suggests allocating more bandwidth to the devices with worse channel conditions or weaker computation capabilities. Then, a greedy device scheduling algorithm is introduced, which in each step selects the device consuming the least updating time obtained by the optimal bandwidth allocation, until the lower bound begins to increase, meaning that scheduling more devices will degrade the model accuracy. Experiments show that the proposed policy outperforms state-of-the-art scheduling policies under extensive settings of data distributions and cell radius.
\end{abstract}


%
\IEEEpeerreviewmaketitle

\section{Introduction}
\IEEEPARstart{A}{ccording} to Cisco's estimation, nearly 850 zettabytes of data will be generated each year at the network edge by 2021 \cite{cisco}.
These valuable data can bring diverse artificial intelligence (AI) services to end users by leveraging deep learning techniques \cite{goodfellow2016deep}, which are developing rapidly in recent years.
However, training AI models (typically deep neural networks) via conventional centralized training methods requires aggregating all raw data to a central server.
Since uploading raw data via wireless channels can drain the wireless bandwidth and cause privacy issues when the raw data are uploaded to the central server \cite{chen2012data}, it is hardly practical to use conventional centralized training methods in wireless networks \cite{saad2019vision}.

To address the aforementioned issues, researchers have proposed a new distributed model training framework called Federated Learning (FL) \cite{mcmahan2017communication, li2019federated}.
A typical wireless FL system leverages the computation capabilities of multiple end devices, which are coordinated by a central controller, for example a base station (BS), to train a model in an iterative fashion \cite{lim2019federated}.
In each iteration of FL (also known as a round), the participating devices use their local data to update the local models, and then the local models are sent to the BS for global model aggregation.
By updating the model parameters locally, FL leverages both the data and computation capabilities distributed on devices, and hence can reduce the model training latency as well as preserving the data privacy.
Therefore, FL becomes a promising technology for distributed data analysis and model training in wireless networks \cite{park2018wireless, zhu2020toward}, and has been used in many applications, for instance, resource allocation optimization in vehicle-to-vehicle (V2V) communications \cite{samarakoon2018federated} and content recommendations for smartphones \cite{bonawitz2019towards}.

However, implementing FL in real wireless networks encounters several key challenges that have not been fully resolved yet.
Due to the scarce spectrum resources and stringent training latency budget, only a limited number of devices are allowed to upload local models in each round, where the device scheduling policy becomes crucial and can affect the convergence rate of FL in two ways.
On the one hand, in each round, the BS cannot perform the global model aggregation until all scheduled devices have finished updating their local models and uploading the local model updates.
Therefore, straggler devices with limited computation capabilities or bad channel conditions can significantly slow down the model aggregation.
As a result, scheduling more devices leads to a longer latency per round, due to the reduced bandwidth allocated to each scheduled device and a higher probability of having straggler devices.
On the other hand, scheduling more devices increases the convergence rate w.r.t. the number of rounds \cite{stich2019local, li2019convergence}, and can potentially reduce the number of rounds required to attain the same accuracy.
Therefore, if we look at the total training time, which is the number of rounds times the average latency per round, the device scheduling is essential and should be carefully optimized to balance the latency per round and the number of required rounds.
Moreover, the scheduling policy should also adapt itself to the dynamic wireless environment.

Recently, implementing FL in wireless networks has received many research efforts.
To reduce the uploading latency introduced by global model aggregation, novel analog aggregation techniques have been proposed in \cite{zhu2018low, amiri2019machine, yang2020federated}.
For analog aggregation, the scheduled devices concurrently transmit their local models via analog modulation in a wireless multiple-access channel, and thus the BS receives the aggregated model thanks to the waveform-superposition property.
Although the uploading latency can be greatly reduced, stringent synchronization among devices is required.
While for digital transmission based FL, the scheduled devices need to share the limited wireless resources and the resource allocation problems have been studied by a series of work.
The authors of \cite{tran2019federated} adopt TDMA for the MAC layer, and jointly optimize the device CPU frequency, the transmission latency, and the local model accuracy to minimize the weighted sum of training latency and total device energy consumption.
A similar FL system but with FDMA is considered in \cite{yang2019energy}.
On the other hand, the frequency of global aggregation under heterogeneous resource constraints has been optimized in \cite{wang2019adaptive, wang2018edge}.
In \cite{tran2019federated, yang2019energy, wang2019adaptive, wang2018edge}, all devices are involved in each round, which is hardly feasible in practical wireless FL applications due to the limited wireless bandwidth.
Another series of work proposes to use device scheduling to optimize the convergence rate of FL.
A heuristic scheduling policy that jointly considers the channel states and the importance of local updated models, is proposed in \cite{amiri2020update}.
However, the proposed scheduling policy is only evaluated by experiments and the convergence performance cannot be theoretically guaranteed.
{A greedy scheduling policy is proposed in \cite{nishio2019client} that schedules as many devices as possible within a given deadline for each round.
Nevertheless, the deadline is chosen through experiments and can hardly be adapted to dynamic channels and device computation capabilities.}
In \cite{zeng2019energy}, the authors exploit an intuition that the convergence rate of FL increases linearly with the number of scheduled devices, and accordingly an energy-efficient joint bandwidth allocation and scheduling policy is proposed.
{The relation between the number of rounds required to attain a certain accuracy and the scheduling policy is derived in \cite{yang2019scheduling}, and three basic scheduling policies, namely random scheduling, round-robin, and proportional fair, are compared.
However, due to the complicated relation between the number of rounds required to attain a certain accuracy and the per round latency, the convergence rate w.r.t. \emph{the number of rounds} obtained by \cite{yang2019scheduling} cannot be directly transformed into the convergence rate w.r.t. \emph{time}.
The authors of \cite{chen2019joint} jointly optimize the uplink resource block allocation and transmission power to maximize the \emph{asymptotic} convergence performance of FL, while the performance can hardly be guaranteed for latency constrained wireless FL applications.}
Therefore, the model accuracy within certain training time budgets (i.e., convergence rate w.r.t. \emph{time}), which is critical for latency constrained FL applications \cite{zhou2019edge, wang2019edge_2}, has not been addressed yet.

In this paper, we aim to optimize the convergence rate of FL w.r.t. \emph{time} rather than the number of rounds.
Specifically, we formulate a joint bandwidth allocation and scheduling problem to maximize the accuracy of the trained model, and further decouple the problem into two sub-problems, i.e., bandwidth allocation and device scheduling.
For the bandwidth allocation problem, assuming a given set of scheduled devices, the implicit optimal solution that minimizes the latency of the current round is first obtained, and an efficient binary search algorithm is proposed to numerically get the optimal bandwidth allocation and the corresponding round latency.
For the device scheduling problem, by relaxing the objective into minimizing the upper bound of the loss function based on a derived convergence bound that incorporates device scheduling,
{we design a greedy algorithm that adds devices one by one with the shortest updating time to the scheduled devices set}, until the convergence bound begins to increase, meaning that scheduling more devices will reduce the convergence rate.
Our main contributions are summarized as follows.
\begin{itemize}
  \item{{We theoretically bound the impact of the device scheduling in each round, based on which the convergence analysis from \cite{wang2019adaptive} is extended to derive a convergence bound of FL in terms of the number of rounds and the number of scheduled devices.} The bound applies for non-independent and identically distributed (non-i.i.d.) local datasets and an arbitrary number of scheduled devices in each round.}
  \item{The obtained convergence bound quantifies the trade-off between the latency per round and the number of required rounds to attain a fixed accuracy, and thus the device scheduling can be accordingly optimized to maximize the convergence rate of FL.}
  \item{Using the obtained convergence bound, we design a device scheduling policy according to the learned loss function characteristics, gradient characteristics, and system dynamics in real time, in order to minimize the loss function value under a given training time budget.}
  \item{Our experiments show that the optimal number of scheduled devices increases with the non-i.i.d. level of local datasets, and the proposed
  scheduling policy adapts to non-i.i.d. local datasets and can achieve near-optimal performance. Moreover, the proposed scheduling policy outperforms several state-of-the-art scheduling policies in terms of the highest achievable accuracy within the total training time budget.}
\end{itemize}

The remainder of this paper is organized as follows. In Section II, we introduce the system model and formulate the convergence rate optimization problem. We derive a convergence bound and approximately solve the problem in Section III. The experiment results are shown in Section IV and we conclude the paper in Section V.

\section{System Model}

We introduce the basic concepts of FL, the procedure of FL, the latency model, and the problem formulation in this section. The main notations are summarized in Table \ref{notation}.

\begin{table}[!t]
\setlength{\abovecaptionskip}{2pt}
\setlength{\belowcaptionskip}{2pt}
\caption{Summary of Main Notations}
\label{notation}
\begin{center}
\begin{tabular}{c p{0.68\linewidth}}
\hline
\textbf{Notation} & \textbf{Definition}\\
\hline
$\mathcal{M}$; $M$ & Set of devices; size of $\mathcal{M}$\\
$R$; $T$; $K$ & Cell radius; total training time budget; total number of rounds within $T$ \\
$\mathcal{D}_i$; $D_i$; {$d_i$} & Local dataset of device $i$; size of $\mathcal{D}_i$; {batch size of the local update of device $i$}\\
$\mathcal{D}$; $D$ & Global dataset; size of $\mathcal{D}$ \\
$F_i(\vec{w})$; $F(\vec{w})$ & Local loss function of device $i$; global loss function \\
$\vec{w}_{i,k}(j)$ & Local model of device $i$ in the $j$-th local update of the $k$-th round \\
$\vec{w}_{k}^{\vec{\Pi}_{[k]}}$; $\tilde{\vec{w}}$  & Global model in the $k$-th round; the global model that has the minimum loss within $T$ \\
$\Pi_k$ & Scheduling policy of the $k$-th round, i.e., the subset of scheduled devices \\
$t_{i,k}^{\text{cp}}$; $t_{i,k}^{\text{cm}}$; $t_k^{\text{round}}(\Pi)$ & Computation latency; communication latency; round latency under policy $\Pi$\\
$B$; $N_0$ & System bandwidth; noise power density\\
$P_i$; $h_{i,k}$ & Transmit power of device $i$; channel gain of device $i$ in the $k$-th round \\
$\gamma_{i,k}$ & Bandwidth allocation ratio of device $i$ in the $k$-th round \\
$\tau$; $\eta$ & Number of local updates performed  by the scheduled devices between two adjacent global aggregations; learning rate \\
$\rho$; $\beta$; $\delta_i$ & Convexity of $F_i(\vec{w})$; smoothness of $F_i(\vec{w})$; divergence of the gradient $\nabla F_i(\vec{w})$\\
\hline
\end{tabular}
\end{center}
\vspace{-10pt}
\end{table}

We consider an FL system consisting of one BS and $M$ end devices, and the devices are indexed by $\mathcal{M}=\{1,2,\dots,M\}$. Each device $i$ has a local dataset $\mathcal{D}_i=\{\vec{x}_{i,d}\in \mathbb{R}^s, y_{i,d} \in\mathbb{R}\}_{d=1}^{D_i}$, with $D_i=|\mathcal{D}_i|$ data samples.
Here $\vec{x}_{i,d}$ is the $d$-th $s$-dimensional input data vector at device $i$, and $y_{i,d}$ is the labeled output of $\vec{x}_{i,d}$.
The whole dataset is denoted by $\mathcal{D} = \mathop{\cup} \limits_{i \in \mathcal{M}}\mathcal{D}_i$ with total number of samples $D = \sum \limits_{i \in \mathcal{M}}D_i$.

The goal of the training process is to find the model parameter $\vec{w}$, so as to minimize a particular loss function on the whole dataset. The optimizing objective can be expressed as
\begin{equation}
\setlength\abovedisplayskip{3pt}
\setlength\belowdisplayskip{3pt}
    \min \limits_{\vec{w}} \left\{ F(\vec{w}) \triangleq
     \frac{1}{D}\sum_{i\in \mathcal{M}} D_i F_i(\vec{w})\right\}, \label{GL}
\end{equation}
where the local loss function $F_i(\vec{w})$ is defined as $F_i(\vec{w}) \triangleq \frac{1}{D_i}\sum_{\{\vec{x}_{i,d}, y_{i,d}\} \in \mathcal{D}_i} f(\vec{w}, \vec{x}_{i,d}, y_{i,d})$,
and the loss function $f(\vec{w}, \vec{x}_{i,d}, y_{i,d})$ captures the error of the model parameter $\vec{w}$ on the input-output data pair $\{\vec{x}_{i,d}, y_{i,d}\}$. Some examples of loss functions used in popular machine learning models are summarized in Table \ref{table1}.

\begin{table}[!t]
\setlength{\abovecaptionskip}{2pt}
\setlength{\belowcaptionskip}{2pt}
\caption{Loss Functions for Popular Machine Learning Models}
\label{table1}
\begin{center}
\begin{tabular}{c p{0.6\linewidth}}
\hline
\textbf{Model} & \textbf{Loss function $f(\vec{w}, \vec{x}_i, y_i)$} \\
\hline
Linear regression & $\frac{1}{2}\left\| y_i- \vec{w}^\mathsf{T}\vec{x}_i \right\|^2$\\
Squared-SVM & $\frac{\lambda}{2}\left\| \vec{w} \right\|^2+\frac{1}{2}\text{max}\{0; 1-y_i\vec{w}^\mathsf{T}\vec{x}_i\}$, where $\lambda$ is a constant \\
Neural network & Cross-entropy on cascaded linear and non-linear transform, see \cite{goodfellow2016deep} for details\\
\hline
\end{tabular}
\end{center}
\vspace{-10pt}
\end{table}

\subsection{Federated Learning over Wireless Networks}
FL uses an iterative approach to solve problem \eqref{GL},
and each round, indexed by $k$, contains the following 3 steps.
\begin{enumerate}
    \item The BS first decides to schedule which devices to participate in the current round, and the set of scheduled devices in round $k$ is denoted by $\Pi_k$. Then the BS broadcasts the current global model $\vec{w}_{k-1}^{\vec{\Pi}_{[k-1]}}$ to all scheduled devices, where $\vec{\Pi}_{[k-1]} \triangleq [\Pi_1, \Pi_2, \dots, \Pi_{k-1}]$ denotes the historical scheduling decisions up to the $(k-1)$-th round.
    \item Each scheduled device $i\in \Pi_k$ receives the global model (i.e., $\vec{w}_{i,k}(0)\gets \vec{w}_{k-1}^{\vec{\Pi}_{[k-1]}}$) and
    updates its local model by applying the gradient descent algorithm on its local dataset:
    \begin{equation}
    \setlength\abovedisplayskip{3pt}
    \setlength\belowdisplayskip{3pt}
        \vec{w}_{i,k}(j+1)=\vec{w}_{i,k}(j)-\eta \nabla F_i(\vec{w}_{i,k}(j))   ,\  j=0,1,\dots,\tau-1, \label{localupdates}
    \end{equation}
    where $\eta$ is the learning rate.
    {In practice, the local dataset may have thousands or even millions of data samples, making the gradient descent impractical. Therefore, stochastic gradient descent (SGD), which can be regarded as a stochastic approximation of gradient descent, is widely used as a substitution.
    In SGD, the gradient $\nabla F_i(\vec{w}_{i,k}(j))$ is computed on $\mathcal{D}_{b,i}$, a randomly sampled subset from $\mathcal{D}_i$, where $\mathcal{D}_{b,i}$ is called mini-batch and $d_i = |\mathcal{D}_{b,i}|$ is called batch size.}
    The local model update is repeated for $\tau$ times and $\tau$ is considered as a fixed system parameter.
    Then the updated local model $\vec{w}_{i,k}(\tau)$ is uploaded to the BS. In the following part of the paper, we use $\vec{w}_{i,k}$ to denote $\vec{w}_{i,k}(\tau)$ unless otherwise specified.
    \item After receiving all the uploaded models, the BS aggregates them (i.e., weighted averages the uploaded local models according to the size of local datasets) to obtain a new global model:
    \begin{equation}
    \setlength\abovedisplayskip{3pt}
    \setlength\belowdisplayskip{3pt}
        \vec{w}_{k}^{\vec{\Pi}_{[k]}} = \frac{\sum_{i\in\Pi_k} D_i\vec{w}_{i,k}}{\sum_{i\in\Pi_k}D_i}. \label{globalaggregation}
    \end{equation}
\end{enumerate}

\subsection{Latency Model}
We consider an arbitrary round $k$, the total latency of the $k$-th round consists of the following parts:
\subsubsection{Computation Latency}
To characterize the randomness of the computation latency of local model update, we use the shifted exponential distribution\cite{lee2017speeding, li2016unified}:
\begin{equation}
\setlength\abovedisplayskip{3pt}
\setlength\belowdisplayskip{3pt}
\mathbb{P}[t_{i,k}^{\text{cp}}<t]=
\begin{cases}
{1-e^{-\frac{ \mu_i}{\tau d_i}(t-a_i \tau d_i)}} & {\text{, $t\geq a_i \tau d_i$,}} \\
0 & \text{, otherwise,}
\end{cases}
\label{shifted_exponential}
\end{equation}
where $a_i>0$ and $\mu_i>0$ are parameters that indicate the maximum and fluctuation of the computation capabilities, respectively. We assume that $a_i$ and $\mu_i$ stay constant throughout the whole training process.
Moreover, we ignore the computation latency of the model aggregation at the BS, due to the relatively stronger computation capability of the BS and low complexity of the model aggregation.

\subsubsection{Communication Latency}
Regarding the local model uploading phase of the scheduled devices, we consider an FDMA system with total bandwidth $B$. The bandwidth allocated to device $i$ is denoted by $\gamma_{i,k} B$, where $\gamma_{i,k}$ is the allocation ratio that satisfies $\sum_{i=1}^{M} \gamma_{i,k} \leq 1$ and $0 \leq \gamma_{i,k} \leq 1$. Therefore, the achievable transmission rate (bits/s) can be written as
$r_{i,k}  = \gamma_{i,k} B \text{log}_2\left(1+\frac{P_i h_{i,k}^2}{\gamma_{i,k} B N_0}\right)$,
where $P_i$ denotes the transmit power of device $i$, that stays constant with different rounds, and $h_{i,k}$ denotes the corresponding channel gain, and $N_0$ is the noise power density. Thus the communication latency of device $i$ is
\begin{equation}
\setlength\abovedisplayskip{3pt}
\setlength\belowdisplayskip{3pt}
  t_{i,k}^{\text{cm}} = \frac{S}{r_{i,k}},
\end{equation}
where $S$ denotes the size of $\Vec{w}_{i,k}$, in bits.
Since the transmit power of the BS is much higher than that of the devices and the whole downlink bandwidth is used by BS to broadcast the model, here we ignore the latency of broadcasting the global model.

Due to the synchronous model aggregation of FL, the total latency per round $t_k^{\text{round}}(\Pi_k)$ is determined by the slowest device among all the scheduled devices, i.e.,
\begin{align}
\setlength\abovedisplayskip{3pt}
\setlength\belowdisplayskip{3pt}
    t_k^{\text{round}}(\Pi_k) \geq \max_{i\in\Pi_k} \{t_{i,k}^{\text{cm}}+t_{i,k}^{\text{cp}}\}.
\end{align}

\subsection{Problem Formulation}
A joint bandwidth allocation and scheduling problem is formulated  to optimize the convergence rate of FL w.r.t. \emph{time}.
Specifically, we use $K$ to denote the total number of rounds within the training time budget $T$, and minimize the global loss function of $\tilde{\vec{w}}$ within $T$, where $\tilde{\vec{w}}$ is the optimal model parameter that has the minimum global loss function value in the whole training process and defined as
\begin{equation}
\setlength\abovedisplayskip{3pt}
\setlength\belowdisplayskip{3pt}
    \tilde{\vec{w}} \triangleq \underset{\vec{w}\in \{ \vec{w}_k^{\vec{\Pi}_{[k]}}:k=1,2,\dots,K\}}{\text{arg\,min}} F(\vec{w}).
    \label{def-tildew}
\end{equation}
For simplicity, we use $[K]$ and $[M]$ to denote $\{1, 2, \dots, K\}$ and $\{1, 2, \dots, M\}$, respectively. The optimization problem can be written as follows:
\begin{align}
\setlength\abovedisplayskip{3pt}
\setlength\belowdisplayskip{3pt}
    \underset{K, \Vec{\Pi}_{[K]}, \Vec{\gamma}_{[K]}, \vec{t}^\text{round}_{[K]}}{\text{min}} \quad & F(\tilde{\vec{w}}) \tag{P1}\label{P1}\\
    \text{s.t.} \qquad \;\; \quad & \sum_{k=1}^K t^{\text{round}}_k(\Pi_k) \leq T,  \tag{C1.1}\label{C11}\\
    & t_{i,k}^\text{cp}+\frac{S}{\gamma_{i,k}B\text{log}_2\left(1+\frac{P_ih_{i,k}^2}{\gamma_{i,k}BN_0}\right)}\leq t^{\text{round}}_k(\Pi_k),  \tag{C1.2}\label{C12} \\
    & \Pi_k \subset \mathcal{M}, \forall k \in [K], \tag{C1.3}\label{C13} \\
    & \sum_{i=1}^{M} \gamma_{i,k} \leq 1, \forall k \in [K], \tag{C1.4}\label{C14} \\
    &  0 \leq \gamma_{i,k} \leq 1,  \forall k \in [K], \forall i \in[M], \tag{C1.5}\label{C15}
\end{align}
where $\vec{\Pi}_{[K]} = [\Pi_1, \Pi_2, \dots, \Pi_K]$, $\vec{\gamma}_{[K]} \triangleq [\vec{\gamma}_1, \vec{\gamma}_2, \dots, \vec{\gamma}_K]$
with $\vec{\gamma}_k \triangleq [\gamma_{1,k}, \gamma_{2,k}, \dots, \gamma_{M,k}]$, \- and $\vec{t}^\text{round}_{[K]} \triangleq [t^\text{round}_1(\Pi_1), t^\text{round}_2(\Pi_2), \dots, t^\text{round}_K(\Pi_K)]$.

To solve \ref{P1}, we need to know how $K$ and $\Vec{\Pi}_{[K]}$ affect the loss function of the final global model, i.e., $F(\tilde{\vec{w}})$.
Since it is almost impossible to find an exact analytical expression of $F(\tilde{\vec{w}})$ w.r.t. $K$ and $\Vec{\Pi}_{[K]}$, we turn to bound $F(\tilde{\vec{w}})$ in terms of $K$ and $\Vec{\Pi}_{[K]}$.
While in our problem, the local computation latency $t_{i,k}^\text{cp}$ and wireless channel state $h_{i,k}$ can vary with different $k$, thus the optimal scheduling policy $\Vec{\Pi}^*_{[K]}$ can be non-stationary.
Moreover, due to the iterative nature of FL, the global model is related to the scheduling policies of all past rounds.
\if
Although the convergence property of FL has recently received great research efforts, only upper bounds on the convergence rate for some deterministic scheduling policies are known \cite{wang2019adaptive}\cite{yang2019scheduling}.
Therefore, it is generally impossible to find an exact analytical expression of $F(\tilde{\vec{w}})$ w.r.t. $K$ and a non-stationary $\Vec{\Pi}_{[K]}$.
Moreover, the global model of the $(k+1)$-th round (i.e., $\vec{w}_{k+1}^{\vec{\Pi}_{[k+1]}}$) is related to the scheduling policy of the $k$-th round (i.e., $\Pi_k$) and the local updated models of the scheduled devices (i.e., $\vec{w}_{i,k}, \forall i \in \Pi_k$) according to \eqref{globalaggregation}. And the local updated models are related to the global model of the $k$-th round according to \eqref{localupdates}.
Therefore, the global model is related to the scheduling policies of all past rounds.
\fi
As a result, it is very hard to bound  $F(\tilde{\vec{w}})$ under a non-stationary scheduling policy.

In the next section, \ref{P1} is solved in the following way.
First, we decouple \ref{P1} into two sub-problems, namely device scheduling and bandwidth allocation.
Then given the scheduled devices, the bandwidth allocation problem is analytically solved.
Further, based on the optimal bandwidth allocation, and a derived convergence bound of FL under a stationary random scheduling policy, we approximately solve the device scheduling problem with a joint device scheduling and bandwidth allocation algorithm.

\section{Joint Device Scheduling and Bandwidth Allocation}
\ref{P1} is decoupled as follows. First, given the scheduling policy of the $k$-th round (i.e., $\Pi_k$), the bandwidth allocation problem of the $k$-th round can be written as follows:
\begin{align}
\setlength\abovedisplayskip{3pt}
\setlength\belowdisplayskip{3pt}
    \underset{\gamma_{i,k}, t^\text{round}_k(\Pi_k)} {\text{min}} \quad &  t^\text{round}_k(\Pi_k) \tag{P2}\label{P2} \\
    \text{s.t.} \qquad \,\,\  & t_{i,k}^\text{cp}+\frac{S}{\gamma_{i,k}B\text{log}_2\left(1+\frac{P_ih_{i,k}^2}{\gamma_{i,k}BN_0}\right)}\leq t^{\text{round}}_k(\Pi_k), \tag{C2.1}\label{C21} \\
    & \sum_{i=1}^{M} \gamma_{i,k} \leq 1,  \tag{C2.2}\label{C22} \\
    &  0 \leq \gamma_{i,k} \leq 1,   \forall i \in[M]. \tag{C2.3}\label{C23}
\end{align}
Then we denote the optimal value of $t^\text{round}_k(\Pi_k)$ as $t^*_k(\Pi_k)$, the device scheduling problem can be written as follows:
\begin{align}
\setlength\abovedisplayskip{3pt}
\setlength\belowdisplayskip{3pt}
    \underset{K, \Vec{\Pi}_{[K]}} {\text{min}} \quad &  F(\tilde{\vec{w}}) \tag{P3}\label{P3} \\
    \text{s.t.} \quad\,\,\ & \sum_{k=1}^K t^*_k(\Pi_k) \leq T, \tag{C3.1}\label{C31} \\
    & \Pi_k \subset \mathcal{M}. \forall k \in [K]. \tag{C3.2}\label{C32}
\end{align}

\subsection{Bandwidth Allocation}
The optimal solution of \ref{P2} can be obtained using the following theorem.
\begin{theorem}\label{thm3}
The optimal bandwidth allocation of \ref{P2} is as follows
\begin{equation}
\setlength\abovedisplayskip{3pt}
\setlength\belowdisplayskip{3pt}
   \gamma^*_{i,k} = \frac{S \rm{ln}2}{\left(t^*_k(\Pi_k)-t_{i,k}^{\rm{cp}}\right)\left(W\left(-\Gamma_{i,k} e^{-\Gamma_{i,k}}\right)+\Gamma_{i,k}\right)},
   \label{thm1eq1}
\end{equation}
where $\Gamma_{i,k} \triangleq \frac{N_0S\rm{ln}2}{\left(t^*_k(\Pi_k)-t_{i,k}^{\rm{cp}}\right)P_ih_{i,k}^2}$ , $W(\cdot)$ is Lambert-W function, and $t^*_k(\Pi_k)$ is the objective value of (P2) that satisfies
\begin{equation}
\setlength\abovedisplayskip{3pt}
\setlength\belowdisplayskip{3pt}
   \sum_{i\in\Pi_k} \gamma^*_{i,k} = \sum_{i\in\Pi_k} \frac{S \rm{ln}2}{\left(t^*_k(\Pi_k)-t_{i,k}^{\rm{cp}}\right)\left(W\left(-\Gamma_{i,k} e^{-\Gamma_{i,k}}\right)+\Gamma_{i,k}\right)}  = 1.
   \label{t*}
\end{equation}
\end{theorem}
\begin{proof}
See Appendix \ref{appendix3}.
\end{proof}

Due to the Lambert-W function in \eqref{t*}, in which the argument is related to $t^*_k(\Pi_k)$ via $\Gamma_{i,k}$, we cannot analytically solve \eqref{t*} to derive $t^*_k(\Pi_k)$. Thus a binary search algorithm (Alg. \ref{bandwidthallocation}) is proposed to get the optimal value of \ref{P2} numerically.
Begin with the target value $t$ that equals to the upper bound of the initial searching region $[t_{\text{low}}, t_{\text{up}}]$,
{we iteratively compute the required bandwidth for the current target value $t$ by substituting $t^*_k(\Pi_k)=t$ into \eqref{thm1eq1} (step 3), and derive the total required bandwidth allocation ratio (step 4). The searching region is halved and the smaller half will be retained if the bandwidth is surplus (steps 7-8), while the larger half will be retained if the bandwidth is deficit (steps 9-10).
The searching terminates when the given precision requirement (i.e., $\varepsilon$) is satisfied (steps 5-6), and thus the complexity of Alg. \ref{bandwidthallocation} is on the order of $\mathcal{O}\left(|\Pi_k|\text{log}_2\left(\frac{t_\text{up}}{\varepsilon}\right)\right)$.}

\begin{algorithm}[!t]
    \caption{Binary Search for the Objective Value of \ref{P2}}
    \label{bandwidthallocation}
    \begin{algorithmic}[1]
    \STATE {Give a big enough $t_\text{up}$, initialize $t_\text{low} = \underset{i \in \Pi_k}{\text{max}}\{t_{i,k}^\text{cp}\}$, $t = t_\text{up}$, and set $\mathrm{success = False}$}
    \WHILE {NOT $\mathrm{success}$}
        \STATE {For each user $i \in \Pi_k$, compute the required bandwidth allocation ratio $\gamma_{i,k}$ using \eqref{thm1eq1} by substituting $t_k^*(\Pi_k) = t$}
        \STATE {Compute the summation of required bandwidth allocation ratio $s = \sum_{i \in \Pi_k} \gamma_{i, k}$ }
        \IF {$1-\varepsilon \leq s \leq 1$}
            \STATE {Obtain the solution with accuracy level $\varepsilon$, set $\mathrm{success = True}$}
        \ELSIF {$0< s < 1-\varepsilon$}
            \STATE {Halve the searching region according to {$t_\text{up} = t$, $t = \frac{t+t_\text{low}}{2}$}}
        \ELSE
            \STATE {Halve the searching region according to {$t_\text{low} = t$, $t = \frac{t+t_\text{up}}{2}$}}
        \ENDIF
    \ENDWHILE
    \RETURN {$t$, and $\gamma_{i,k}, \forall i \in \Pi_k$}
    \end{algorithmic}
\end{algorithm}

\subsection{Convergence Analysis}

Before the convergence analysis, we first introduce some notations, as shown in Fig. \ref{fg1}.
For the stationary random scheduling policy $\Pi$, we use $\vec{w}_k^\Pi$ to denote $\vec{w}_{k}^{\vec{\Pi}_{[k]}}$.
Two auxiliary model parameter vectors are introduced, where $\vec{w}_k$ ($k\geq 1$) is used to denote the model parameter vector that is synchronized with $\vec{w}^\Pi_{k-1}$ at the beginning of the $k$-th round, and is updated by scheduling \emph{all devices} (i.e., $\Vec{w}_k\triangleq\frac{\sum_{i\in\mathcal{M}} D_i\vec{w}_{i,k}}{\sum_{i\in\mathcal{M}}D_i}$) in the $k$-th round. While $\vec{v}_k$ ($k\geq 1$) is used to denote the model parameter vector that is synchronized with $\vec{w}^\Pi_{k-1}$ at the beginning of the $k$-th round, and is updated by \emph{centralized gradient descent}. In the centralized gradient descent procedure of the $k$-th round, $\vec{v}_k$ is updated according to $\vec{v}_k\gets\vec{v}_k-\eta\nabla F(\vec{v}_k)$ for $\tau$ times.
\begin{figure}[!t]
\setlength{\abovecaptionskip}{2pt}
\setlength{\belowcaptionskip}{2pt}
\centering
\includegraphics[width=0.40\linewidth]{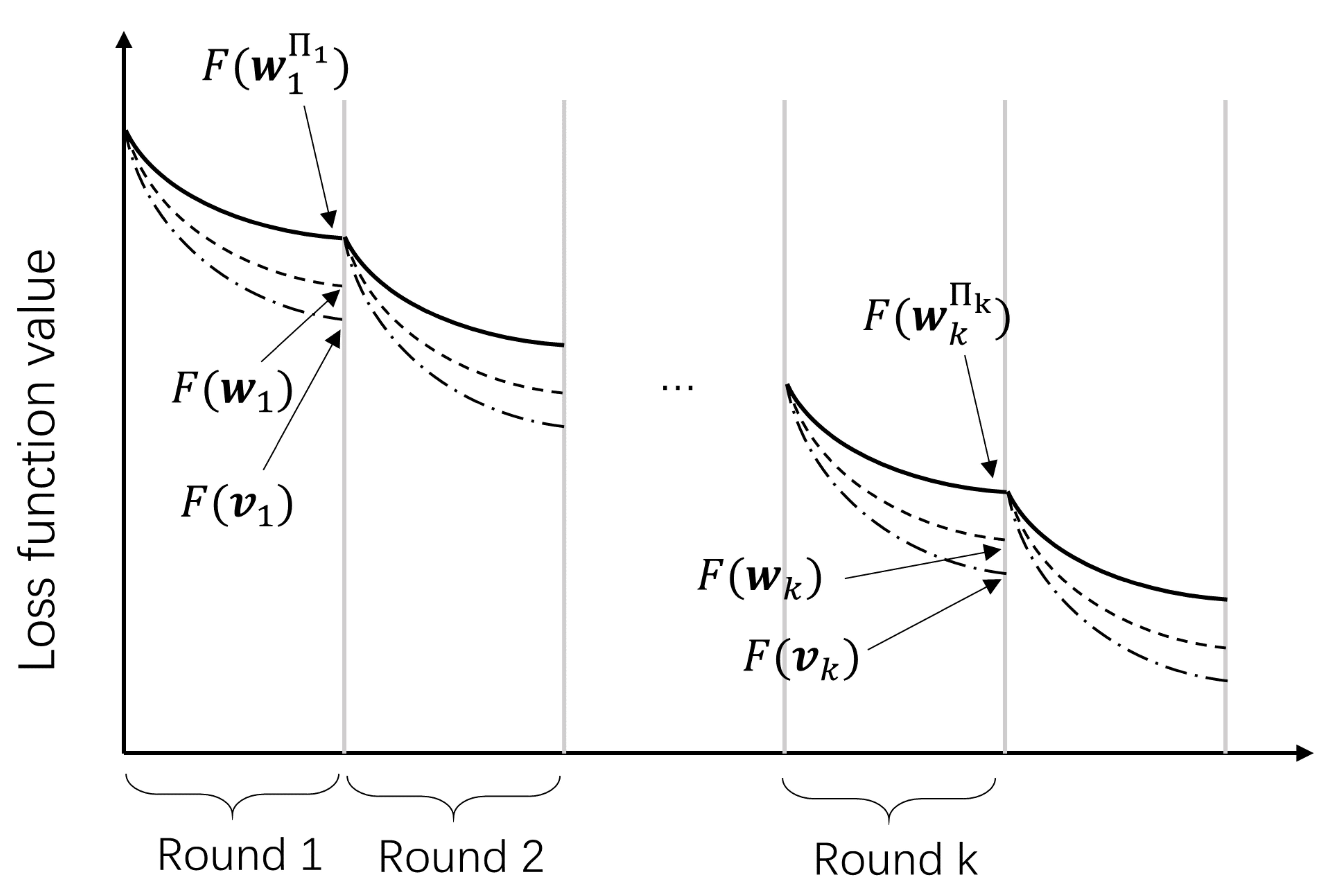}
\caption{Illustration of definitions of different parameter vectors.}
\label{fg1}
\vspace{-10pt}
\end{figure}

To facilitate the analysis, we make the following assumptions on the loss functions $F(\cdot)$.

\begin{assumption}\label{assumption1}
We assume the following for the loss functions of all devices:

\begin{itemize}
    \item $F_i(\vec{w})$ is convex.
    \item $F_i(\vec{w})$ is $\rho$-Lipschitz, i.e., $\norm{F_i(\vec{w})-F_i(\vec{w}')}\leq \rho\norm{\vec{w}-\vec{w}'}$, for any $\vec{w},\vec{w}'$.
    \item $F_i(\vec{w})$ is $\beta$-smooth, i.e., $\norm{\nabla F_i(\vec{w})-\nabla F_i(\vec{w}')}\leq \beta\norm{\vec{w}-\vec{w}'}$, for any $\vec{w},\vec{w}'$.
    \item For any $i$ and $\vec{w}$, the difference between the local gradient and the global gradient can be bounded by $\norm{\nabla F_i(\vec{w})-\nabla F(\vec{w})}\leq \delta_i$, and define $\delta \triangleq \frac{\sum_i D_i \delta_i}{D}$.
\end{itemize}
\end{assumption}

These assumptions are widely used in the literature of convergence analysis for FL \cite{wang2019adaptive, yang2019scheduling,yang2019energy,chen2019joint}, although the loss functions of some machine learning models (e.g., neural network) do not fully satisfy them, especially the convexity assumption. However, our experiment results show that the proposed scheduling policy works well even for the neural network.

To begin with, we derive the upper bound of the difference between the global model aggregated form a stationary random scheduling policy $\Pi$ (i.e., $\vec{w}^\Pi_k$) and $\vec{w}_k$.
\begin{definition}
We define a policy $\Pi$ as a stationary random scheduling policy if and only if $\Pi$ is a size-$|\Pi|$ subset, which is uniformly random sampled from all devices $\mathcal{M}$, and $|\Pi|$ stays constant during the whole training process.
\end{definition}
\begin{theorem}\label{thm1}

For any $k$ and stationary random scheduling policy $\Pi$ ($|\Pi| \geq 1$), we have
\begin{align}
\setlength\abovedisplayskip{3pt}
\setlength\belowdisplayskip{3pt}
    \mathbb{E}\left\{F(\vec{w}^\Pi_k)-F(\vec{w}_k)\right\}
    & \leq \frac{M-\vert\Pi\vert}{|\Pi|} \cdot \underbrace{\frac{\beta\sum_{i=1}^M \sum_{j=1}^M \left( D_i^2 D_j^2 \left(g_i^2(\tau)+g_j^2(\tau)\right)\right)}{2M(M-1)D_{\rm min}^2 D^2} }_{A}  \triangleq B(\Pi), \label{divergence}
\end{align}
where $D_{\rm min}\triangleq {\rm min}_{i\in\mathcal{M}}D_i$, $g_i(x)\triangleq\frac{\delta_i}{\beta}\left((\eta\beta+1)^x-1\right)$, and the expectation is taken over the randomness of $\Pi$.
\end{theorem}
\begin{proof}
See Appendix \ref{appendix1}.
\end{proof}

Note that we always have the learning rate $\eta > 0$, otherwise the gradient descent procedure becomes trivial.
We also have $\beta>0$ and $\delta_i>0$, otherwise the loss function and its gradient become trivial.
Therefore, $g_i(x) > 0$ for $x=1, 2, \dots, \tau$, and thus $A > 0$, where $A$ is defined in \eqref{divergence}.
It is obvious that $A$ is not related to $\Pi$, and $\frac{M-|\Pi|}{|\Pi|}$ decreases with $|\Pi|$.
Therefore, scheduling fewer devices leads to a larger upper bound of $\mathbb{E}\left\{F(\vec{w}^\Pi_k)-F(\vec{w}_k)\right\}$, as thus, a larger upper bound of $\mathbb{E}\left\{F(\vec{w}^\Pi_k)\right\}$.
This means that scheduling fewer devices slows down the convergence rate w.r.t. the number of rounds, which is consistent with conclusions from existing work \cite{zeng2019energy, nishio2019client}.
Furthermore, when $\Pi=\mathcal{M}$ (i.e., schedule all devices), $B(\Pi)$ achieves its lower bound zero, which is consistent with the definition of $\vec{w}_k$.,

{Then, we can combine Theorem \ref{thm1} with the convergence analysis in \cite{wang2019adaptive} to derive the following theorem, which bounds the difference between $\tilde{\vec{w}}$ and $\vec{w}^*$.
In Theorem \ref{thm2}, $\tilde{\vec{w}}$, defined in \eqref{def-tildew}, is the optimal model parameter that has the minimum global loss function value in the whole training process, and $\vec{w}^*$ is the true optimal model parameter that minimizes $F(\vec{w})$.}

\begin{theorem}\label{thm2}
When $\eta\leq\frac{1}{\beta}$ and $\Pi$ is a stationary random scheduling policy, the difference between $F(\tilde{\vec{w}})$ and $F(\vec{w}^*)$ satisfies:
\begin{equation}
\setlength\abovedisplayskip{3pt}
\setlength\belowdisplayskip{3pt}
    \mathbb{E}\left\{\frac{1}{F(\tilde{\vec{w}})-F(\vec{w}^*)}\right\} \geq \frac{1}{\epsilon_0 + \rho h(\tau) + B(\Pi)},
    \label{eqthm2}
\end{equation}
where $\epsilon_0 \triangleq \frac{1+\sqrt{1+4\eta\varphi K^2\tau\left(\rho h(\tau)+B(\Pi)\right)}}{2\eta\varphi K \tau}$,
$\varphi \triangleq \omega\left(1-\frac{\beta\eta}{2}\right)$, $\omega\triangleq {\rm{min}}_k\frac{1}{\norm{\vec{w}^\Pi_k-\vec{w}^*}}$,
$h(x) \triangleq \frac{\delta}{\beta}((\eta\beta+1)^x-1)- \eta\delta x$, and the expectation is taken over the randomness of $\Pi$.
\end{theorem}
\begin{proof}
See Appendix \ref{appendix2}.
\end{proof}

Theorem \ref{thm2} quantifies the trade-off between the latency per round and the number of required rounds.
Scheduling more devices increases the latency per round, and thus decreases the number of possible rounds within the given training time budget $T$ (i.e., $K$), while a smaller $K$ can decrease the lower bound of $\mathbb{E}\left\{\frac{1}{F(\tilde{\vec{w}})-F(\vec{w}^*)}\right\}$.
At the same time, scheduling more devices decreases the value of $B(\Pi)$ as shown by Theorem \ref{thm1}, while a smaller $B(\Pi)$ can increase the lower bound of $\mathbb{E}\left\{\frac{1}{F(\tilde{\vec{w}})-F(\vec{w}^*)}\right\}$.
As a result, the scheduling policy should be carefully optimized to balance the trade-off between the latency per round and the number of required rounds, in order to minimize the loss function of the optimal global model (i.e., $F(\tilde{\vec{w}})$).

\subsection{Device Scheduling Algorithm}
In real wireless networks, the local computation latency $t_{i,k}^\text{cp}$ and wireless channel state $h_{i,k}$ can vary in different rounds $k$, due to the fluctuation of the wireless channels and device computation capabilities. Therefore, at the $k$-th round, $t_{i,k'}^\text{cp}$ and $h_{i,k'}$ for $k'>k$ are unknown, making the constraint \eqref{C31} in \ref{P3} intractable because of the unknown $t^*_{k'}(\Pi_{k'})$ for $k'>k$.
To address this issue, we solve \ref{P3} myopically.
Consider an arbitrary round $k$ and an arbitrary scheduling policy $\Pi_k$, we approximately view that $\Pi_k$ is used in the whole training process, and thus the number of total rounds can be approximated by $\hat{K} = \left\lfloor\frac{T}{t^*_k(\Pi_k)}\right\rfloor$, where $\lfloor \cdot \rfloor$ denotes floor function.
Furthermore, for a given global loss function, $F(\vec{w}^*)$ is a constant, and thus minimizing $F(\tilde{\vec{w}})$ is equivalent to maximizing $\frac{1}{F(\tilde{\vec{w}})-F(\vec{w}^*)}$.
Since the learning rate $\eta$ can be chosen small enough to satisfy $\eta\leq\frac{1}{\beta}$, the objective of \ref{P3} can be approximated by maximizing the lower bound of $\mathbb{E}\left\{\frac{1}{F(\tilde{\vec{w}})-F(\vec{w}^*)}\right\}$ according to Theorem 3, which is equivalent to minimizing the denominator of the right hand side of \eqref{eqthm2}.
Consequently, \ref{P3} can be approximated by the following myopic problem in each round:
\begin{align}
\setlength\abovedisplayskip{3pt}
\setlength\belowdisplayskip{3pt}
    \underset{\Pi_k} {\text{min}} \quad & \frac{1+\sqrt{1+4\eta\varphi \hat{K}^2\tau\left(\rho h(\tau)+B(\Pi_k)\right)}}{2\eta\varphi \hat{K} \tau} + \rho h(\tau) +B(\Pi_k)  \tag{P4}\label{P5} \\
    \text{s.t.} \quad \, & \hat{K} = \left\lfloor\frac{T}{t^*_k(\Pi_k)}\right\rfloor, \tag{C4.1}\label{C51} \\
    & \Pi_k \subset \mathcal{M}. \tag{C4.2}\label{C52}
\end{align}
\if
Therefore, \ref{P3} can be approximated by the myopic problem in each round:
\begin{align}
    \underset{\Pi_k} {\text{min}} \quad &  F(\tilde{\vec{w}})  \tag{P4}\label{P4} \\
    \text{s.t.} \quad\, &
    \hat{K} = \left\lfloor\frac{T}{t^*_k(\Pi_k)}\right\rfloor, \tag{C4.1}\label{C41} \\
    & \Pi_k \subset \mathcal{M}. \tag{C4.2}\label{C42}
\end{align}
For a given global loss function, $F(\vec{w}^*)$ is a constant, and thus minimizing $F(\tilde{\vec{w}})$ is equivalent to maximizing $\frac{1}{F(\tilde{\vec{w}})-F(\vec{w}^*)}$.
Furthermore, the learning rate $\eta$ can be chosen small enough to satisfy $\eta\leq\frac{1}{\beta}$.
As a result, the objective of \ref{P4} can be approximated by maximizing the lower bound of $\mathbb{E}\left\{\frac{1}{F(\tilde{\vec{w}})-F(\vec{w}^*)}\right\}$ according to Theorem 3, which is equivalent to minimizing the denominator of the right hand side of \eqref{eqthm2}:
\begin{align}
    \underset{\Pi_k} {\text{min}} \quad & \frac{1+\sqrt{1+4\eta\varphi \hat{K}^2\tau\left(\rho h(\tau)+B(\Pi_k)\right)}}{2\eta\varphi \hat{K} \tau} + \rho h(\tau) +B(\Pi_k)  \tag{P5}\label{P5} \\
    \text{s.t.} \quad \, & \hat{K} = \left\lfloor\frac{T}{t^*_k(\Pi_k)}\right\rfloor, \tag{C5.1}\label{C51} \\
    & \Pi_k \subset \mathcal{M}. \tag{C5.2}\label{C52}
\end{align}
\fi

\begin{algorithm}[!t]
    \caption{Greedy Scheduling Algorithm}
    \label{greedy}
    \begin{algorithmic}[1]
    \STATE {Initialize $\Pi \gets \emptyset$}
    \STATE {Greedy scheduling: $x \gets \argmin \limits _{i\in \mathcal{M}} t^*(\{i\})$, with $t^*(\cdot)$ given by Alg. \ref{bandwidthallocation}}
    \STATE {{Update $\mathcal{M} \gets \mathcal{M} \setminus \{x\}$, $\Pi \gets \Pi \cup \{x\}$}}
    \STATE {{Estimate $\hat{K} = \left\lfloor\frac{T}{t^*(\{ x\})}\right\rfloor$ and $C=\frac{1+\sqrt{1+4\eta\varphi \hat{K}^2\tau\left(\rho h(\tau)+B(\Pi)\right)}}{2\eta\varphi \hat{K} \tau}+ \rho h(\tau) +B(\Pi)$}}
    \WHILE {$|\mathcal{M}|>0$}
        \STATE {Greedy scheduling: $x \gets \argmin \limits _{i\in \mathcal{M}} t^*(\Pi \cup \{i\}) $, with $t^*(\cdot)$ given by Alg. \ref{bandwidthallocation}}
        \STATE {{Estimate $\hat{K} = \left\lfloor\frac{T}{t^*(\Pi\cup\{ x\})}\right\rfloor$ and $C'=\frac{1+\sqrt{1+4\eta\varphi \hat{K}^2\tau\left(\rho h(\tau)+B(\Pi \cup \{x\})\right)}}{2\eta\varphi \hat{K} \tau}+ \rho h(\tau) + B(\Pi \cup \{x\})$}}
        \IF {$C'> C$}
            \STATE {Break}
        \ELSE
            \STATE {Update $\mathcal{M} \gets \mathcal{M} \setminus \{x\}$, $\Pi \gets \Pi \cup \{x\}$, and $C\gets C'$}
        \ENDIF
    \ENDWHILE
    \RETURN {$\Pi$}
    \end{algorithmic}
\end{algorithm}
\ref{P5} is still a combinatorial optimization problem due to the constraint \eqref{C52}, which is hard to solve. Therefore we propose a greedy algorithm (Alg. \ref{greedy}) to schedule devices.
{In steps 2-3 of Alg. \ref{greedy}, the round latency of scheduling each unscheduled device is given by Alg. \ref{bandwidthallocation}, based on which we choose the device with the minimum latency into the scheduled devices set. Consequently, we initialize the value of the objective function of \ref{P5} in step 4. Then, a similar process is iteratively performed in steps 6-7, until the objective function of \ref{P5} starts to increase or all devices are scheduled.}
The complexity of Alg. \ref{greedy} is on the order of $\mathcal{O}(|\mathcal{M}|^3)$ (because of calling Alg. \ref{bandwidthallocation} for $\mathcal{O}(|\mathcal{M}|^2)$ times), which is much more efficient than the naive brute force search algorithm on the order of $\mathcal{O}(2^{|\mathcal{M}|})$.

However, due to the unknown optimal model $\vec{w}^*$, it is non-trivial to analytically estimate the value of $\varphi$, and thus we treat $\varphi$ as a system parameter that remains fixed throughout the training process. It is shown in the experiments that a fixed $\varphi$ performs well across different system settings like data distributions and cell radius, while the searching for an appropriate value of $\varphi$ is not difficult as well.

\subsection{The Whole Policy}
In this subsection, we propose the complete procedure of the wireless FL with our fast converge scheduling policy (as shown in Alg. \ref{WFL}), which enables the BS to schedule devices in real-time and minimizes the global loss function within the training time budget.
\begin{algorithm}
    \caption{Wireless FL with Fast Converge Scheduling Policy}
    \label{WFL}
    \begin{algorithmic}[1]
    \STATE {Initialize $\vec{w}_0^\Pi$ and $\tilde{\vec{w}}$ as a constant or random vector}
    \STATE {Initialize $t \gets 0$, $\hat{\Vec{\rho}}\gets[\hat{\rho}_1, \hat{\rho}_2, \dots, \hat{\rho}_M]$, $\hat{\Vec{\beta}}\gets[\hat{\beta}_1, \hat{\beta}_2, \dots, \hat{\beta}_M] $, and $\hat{\Vec{\delta}}\gets[\hat{\delta}_1, \hat{\delta}_2, \dots, \hat{\delta}_M] $}
    \FOR{$k=1,2,\dots$}
        \STATE {Estimate $\hat{\rho} = \frac{\sum_{i\in\mathcal{M}}D_i\hat\rho_i}{D}$, $\hat{\beta} = \frac{\sum_{i\in\mathcal{M}}D_i\hat\beta_i}{D}$, and $\hat{\delta} = \frac{\sum_{i\in\mathcal{M}}D_i\hat\delta_i}{D}$}
        \STATE {Call Alg. \ref{greedy} to derive the scheduling policy $\Pi_k$}
        \STATE {Call Alg. \ref{bandwidthallocation} to derive the bandwidth allocation $\gamma_k$ and the optimal round latency $t^*_k(\Pi_k)$}
        \STATE {$t\gets t + t^*_k(\Pi_k)$}
        \IF{$t > T$}
            \BREAK
        \ENDIF
        \STATE {The BS broadcasts the global model $\vec{w}_{k-1}^\Pi$ to all scheduled devices}
        \FOR{each scheduled device $i\in \Pi_k$ in parallel}
            \STATE {Receive $\vec{w}_{k-1}^\Pi$ and set $\vec{w}_{i,k}(0) \gets \vec{w}_{k-1}^\Pi$}
            \STATE {Perform local model update for $\tau$ times according to \eqref{localupdates}}
            \STATE {Estimate $\hat{\rho}_i = \frac{\norm{F_i(\vec{w}_{k-1}^\Pi)-F_i(\vec{w}_{i,k})}}{\norm{\vec{w}_{k-1}^\Pi-\vec{w}_{i,k}}}$, and $\hat{\beta}_i = \frac{\norm{\nabla F_i(\vec{w}_{k-1}^\Pi)- \nabla F_i(\vec{w}_{i,k})}}{\norm{\vec{w}_{k-1}^\Pi-\vec{w}_{i,k}}}$}
            \STATE {Send $\vec{w}_{i,k}$, $\hat{\rho}_i$, $\hat{\beta}_i$, and $ F_i(\vec{w}_{k-1}^\Pi)$ to the BS}
        \ENDFOR
        \STATE {Receive $\vec{w}_{i,k}$ from each scheduled device and update the global model according to \eqref{globalaggregation}}
        \STATE {Receive $\hat{\rho}_i$ and $\hat{\beta}_i$ from each scheduled device and update the corresponding terms in $\hat{\Vec{\rho}}$ and $\hat{\Vec{\beta}}$, respectively}
        \STATE {Estimate $\nabla F_i(\vec{w}_{k-1}^\Pi)=\frac{\vec{w}_{k-1}^\Pi-\vec{w}_{i,k}}{\tau \eta}$}
        \STATE {Compute $\nabla F(\vec{w}_{k-1}^\Pi) = \frac{\sum_{i\in\Pi_k} D_i \nabla F_i(\vec{w}_{k-1}^\Pi)}{\sum_{i\in\Pi_k} D_i}$,  estimate $\hat{\delta_i} = \norm{\nabla F_i(\vec{w}_{k-1}^\Pi)-\nabla F(\vec{w}_{k-1}^\Pi)}$ for each i and update the corresponding term in $\hat{\vec{\delta}}$}
        \STATE {Receive $F_i(\vec{w}_{k-1}^\Pi)$ from each scheduled device and compute $F(\vec{w}_{k-1}^\Pi)=\frac{\sum_{i\in\Pi_k} D_i  F_i(\vec{w}_{k-1}^\Pi)}{\sum_{i\in\Pi_k} D_i}$}
        \IF {$F(\vec{w}_{k-1}^\Pi)<F(\tilde{\vec{w}})$}
            \STATE {$\tilde{\vec{w}} = \vec{w}_{k-1}^\Pi$}
        \ENDIF
    \ENDFOR
    \end{algorithmic}
\end{algorithm}

In Alg. \ref{WFL}, steps 1-2 are the initialization phase, initializing $\vec{w}_0^\Pi$, $\tilde{\vec{w}}$ for the global model, and $\hat{\Vec{\rho}}$, $\hat{\Vec{\beta}}$, $\hat{\Vec{\delta}}$, which are used to record the real-time estimations of the convergence property parameters.
In each round, Alg. \ref{greedy} is called to obtain the device scheduling policy based on the estimated $\hat{\Vec{\rho}}$, $\hat{\Vec{\beta}}$ and $\hat{\Vec{\delta}}$ (step 5).
Then in step 6, Alg. \ref{bandwidthallocation} is called to obtain the optimal bandwidth allocation for scheduled devices and the corresponding round latency {\footnote{{When the devices are updating the local models, they can send pilot signals to the edge server to estimate the channel and inform the edge server of their progress of local computation with low communication overhead. Therefore, perfect information of $t^\text{cp}_{i,k}$ and $h_{i,k}$ is assumed to be known unless otherwise specified.}}}.
We update the accumulated training latency and check if it exceeds the budget $T$ in steps 7-9.
Despite the regular FL local update procedure (steps 13-14), each scheduled device $i\in\Pi_k$ also needs to estimate $\rho_i$ and $\beta_i$ based on the local loss and gradient of $\vec{w}_{k-1}^\Pi$ and $\vec{w}_{i,k}$ according to step 15.
Then in step 16, the updated local models, the estimations, and the loss function values are sent to the BS.
The BS receives the uploaded local models, based on which the global model is updated according to \eqref{globalaggregation} (step 18), and updates the estimation records of $\hat{\Vec{\rho}}$ and $\hat{\Vec{\beta}}$ (step 19).
We update $\hat{\vec{\delta}}$ in a similar way in steps 20-21 and update $\tilde{\vec{w}}$ in steps 22-25.
Note that Alg. \ref{greedy} needs the estimated $\rho_i$, $\beta_i$, and $\delta_i$ for all devices to compute the convergence bound (according to step 7 in Alg. \ref{greedy}, and \eqref{divergence}), while only the devices that have been scheduled in the last round have the up-to-date estimations.
To address this issue, for each device that has not been scheduled in the last round, we use the latest estimation in the past rounds to approximate the up-to-date estimation.
Therefore, $\hat{\Vec{\rho}}$, $\hat{\Vec{\beta}}$, and $\hat{\Vec{\delta}}$ are used to record the estimations and estimate $\hat{\rho}$, $\hat{\beta}$, and $\hat{\delta}$ according to step 4.

In an arbitrary round $k$, the additional computational complexity of Alg. \ref{WFL} at the BS compared to the conventional FL mainly consists of three parts:
1) the computational complexity of Alg. \ref{greedy}, which is $\mathcal{O}(|\mathcal{M}|^3)$;
2) the computational complexity of Alg. \ref{bandwidthallocation}, which is $\mathcal{O}\left(|\Pi_k|\text{log}_2\left(\frac{t_\text{up}}{\varepsilon}\right)\right)$;
3) the computational complexity of maintaining $\hat{\Vec{\rho}}$, $\hat{\Vec{\beta}}$, and $\hat{\Vec{\delta}}$, which is $\mathcal{O}(|\Pi_k|)$.
Because $|\Pi_k|\leq|\mathcal{M}|$, the total additional computation complexity at the BS is $\mathcal{O}(|\mathcal{M}|^3)$ in each round.
While the additional computation complexity at each device $i$ is $\mathcal{O}(1)$ in each round, due to the estimation of $\rho_i$ and $\beta_i$.
For the signaling overhead, compared to the conventional FL, each scheduled device needs to send 3 extra scalars to the BS in each round (i.e., $\hat{\rho}_i$, $\hat{\beta}_i$, and $F_i(\vec{w}_{k-1}^\Pi)$) as shown in step 16, which is negligible compared to sending the high-dimensional local updated model $\vec{w}_{i,k}$.

\section{Experiment Results}
In this section, we evaluate the performance of FL under the proposed scheduling policy.

\subsection{Environment and FL Setups}
Unless otherwise specified, we consider an FL system that consists of $M=20$ devices located in a cell of radius $R=600$ m and a BS located at the center of the cell.
Assume that all devices are uniformly distributed in the cell at the beginning of each round to reflect mobility {\cite{zhu2018low}}.
The wireless bandwidth is $B=20$ MHz, and the path loss exponent is $\alpha = 3.76$.
The transmit power of devices is set to be $P_i = 10$ dBm, and the power spectrum density of the additive Gaussian noise is $N_0 = -114$ dBm/MHz.

{We evaluate the training performance of the proposed policy under two well-known learning tasks, the MNIST dataset \cite{lecun1998gradient} for handwritten digits classification and the CIFAR-10 dataset \cite{krizhevsky2014cifar} for image classification.
The MNIST dataset has 60,000 training images and 10,000 testing images of the 10 digits, and the CIFAR-10 dataset has 50,000 training images and 10,000 images of 10 types of objects.
We accept the common assumption that each device has equal amount of training data samples and the local training datasets are non-overlapping with each other\cite{wang2019adaptive,zhu2018low}.}
Different training data distributions are considered, including \textit{i.i.d.} case and \textit{non-i.i.d.} cases.
{For the i.i.d. dataset, the original training dataset is randomly partitioned into 20 pieces and each device is assigned a piece.
While for the non-i.i.d. cases, the original training dataset is first partitioned into 10 pieces according to the label, and each piece with the same label is then randomly partitioned into $2l$ shards (i.e., $20l$ shards in total).
Finally, each device is assigned $l$ shards with different labels.
The parameter $l$ captures the non-i.i.d. level of local datasets, where smaller $l$ corresponds to a higher non-i.i.d. level.
Following \cite{amiri2020update, sun2019energy}, we train a multilayer perceptron (MLP) model with a single hidden layer with 64 nodes, and use ReLU activation.}
{The mini-batch size is set to be 128 for the local model update, and each scheduled device performs $\tau = 5$ local updates between two adjacent global aggregations.
The learning rate $\eta$ is set to be $0.01$ for MNIST and $0.02$ for CIFAR-10.}
{
The MLP model has 50,816 multiply-and-accumulate (MAC) operations for MNIST.
Assuming that all devices are of the same kind, having maximum CPU frequency of 1 GHz/s and can process one MAC operation in each CPU cycle, and thus we set $a=0.5$ ms/sample and further set $\mu=\frac{1}{a}$ for the computation latency model
\cite{reisizadeh2019coded}.
The total training time budget $T$ is set to be $60$ seconds for MNIST and $200$ seconds for CIFAR-10, and the initial values of $\hat{\rho}_i$, $\hat{\beta}_i$, and $\hat{\delta}_i$ are 1.5, 12, and 2, respectively.}

\subsection{Evaluation of the Fast Converge Scheduling Policy}
\begin{figure}[!t]
\setlength{\abovecaptionskip}{2pt}
\setlength{\belowcaptionskip}{2pt}
\centering
\subfloat[]
{\includegraphics[width=0.42\linewidth]{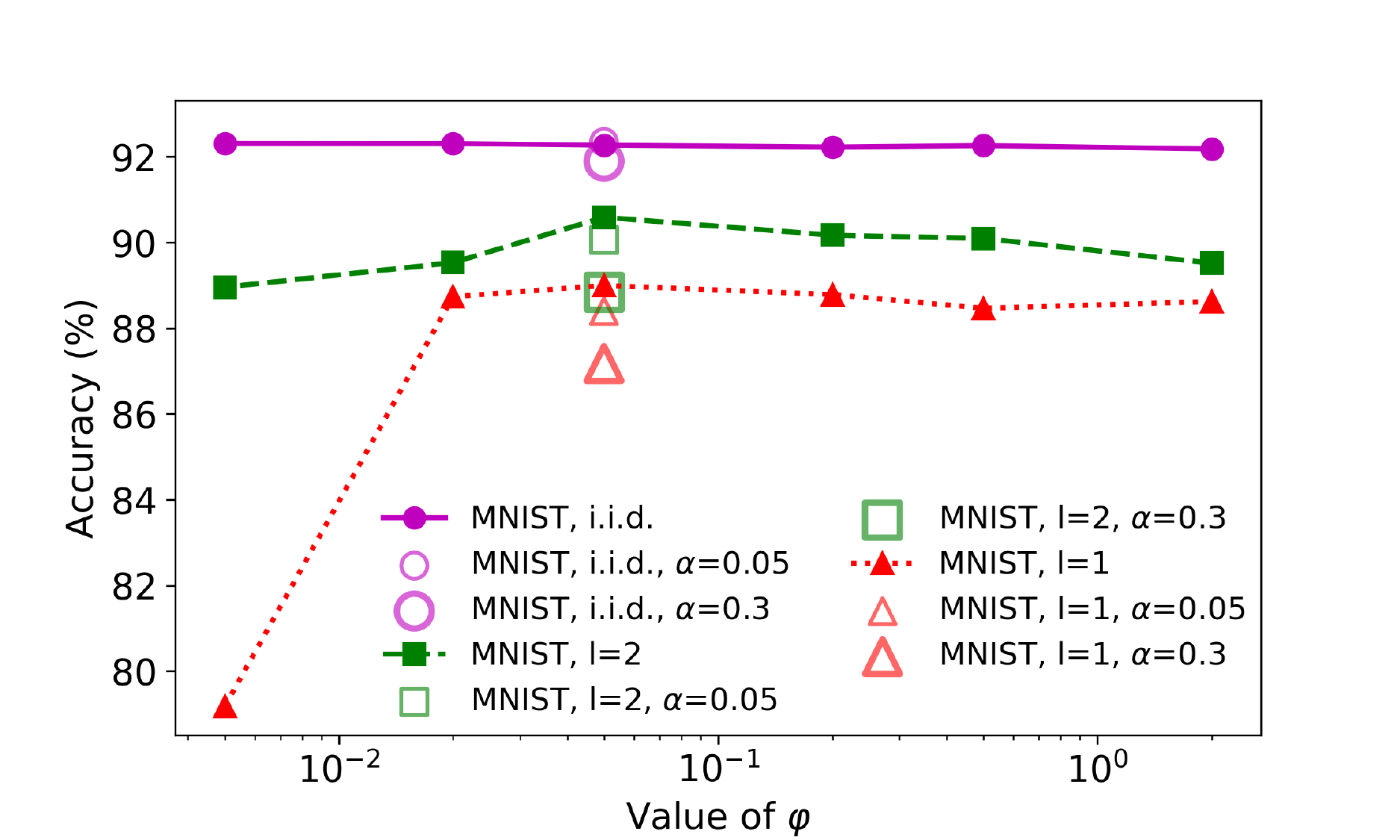}
\label{phi-acc}}
\hfil
\subfloat[]
{\includegraphics[width=0.42\linewidth]{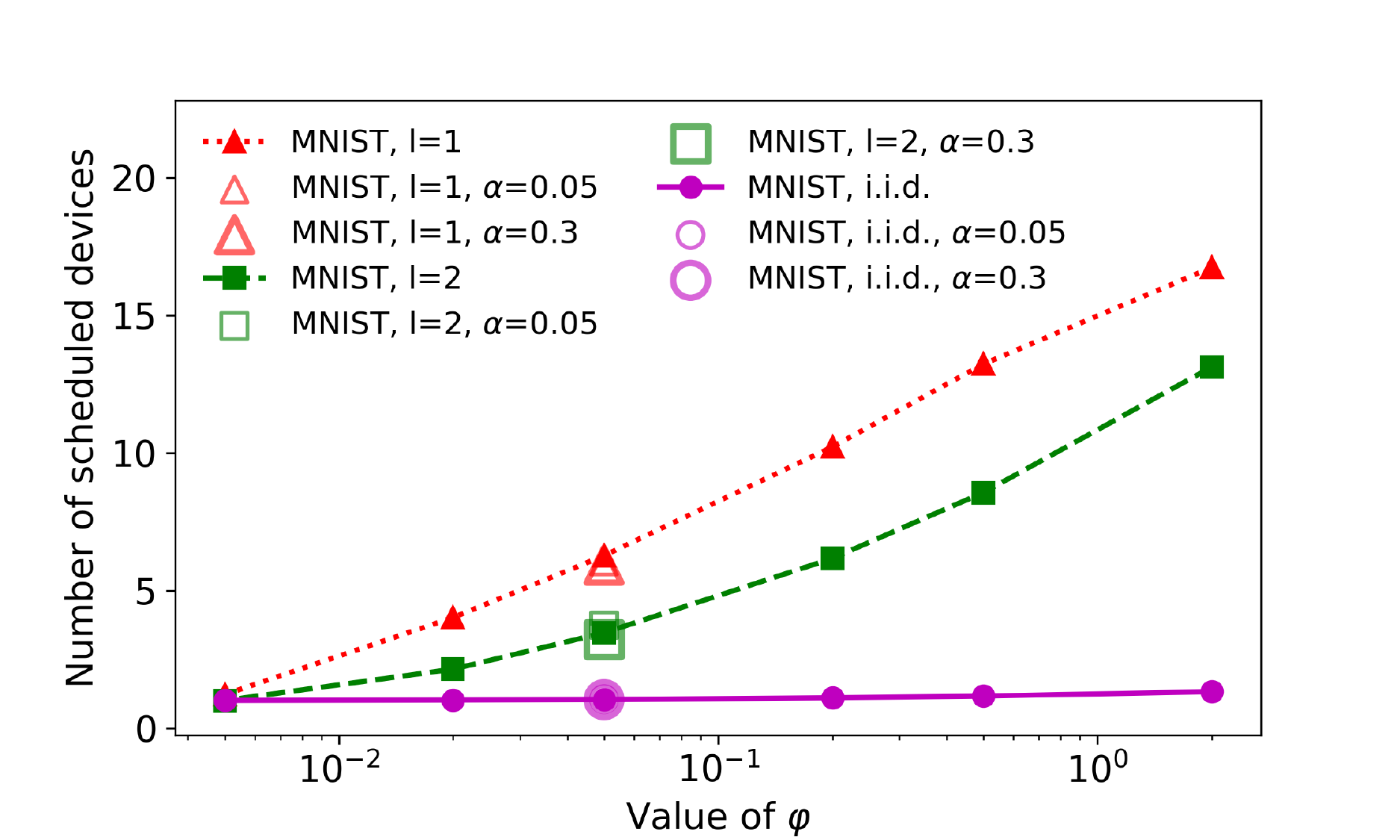}
\label{phi-scheduleddevices}}

\addvspace{-12pt}
\subfloat[]{
\includegraphics[width=0.42\linewidth]{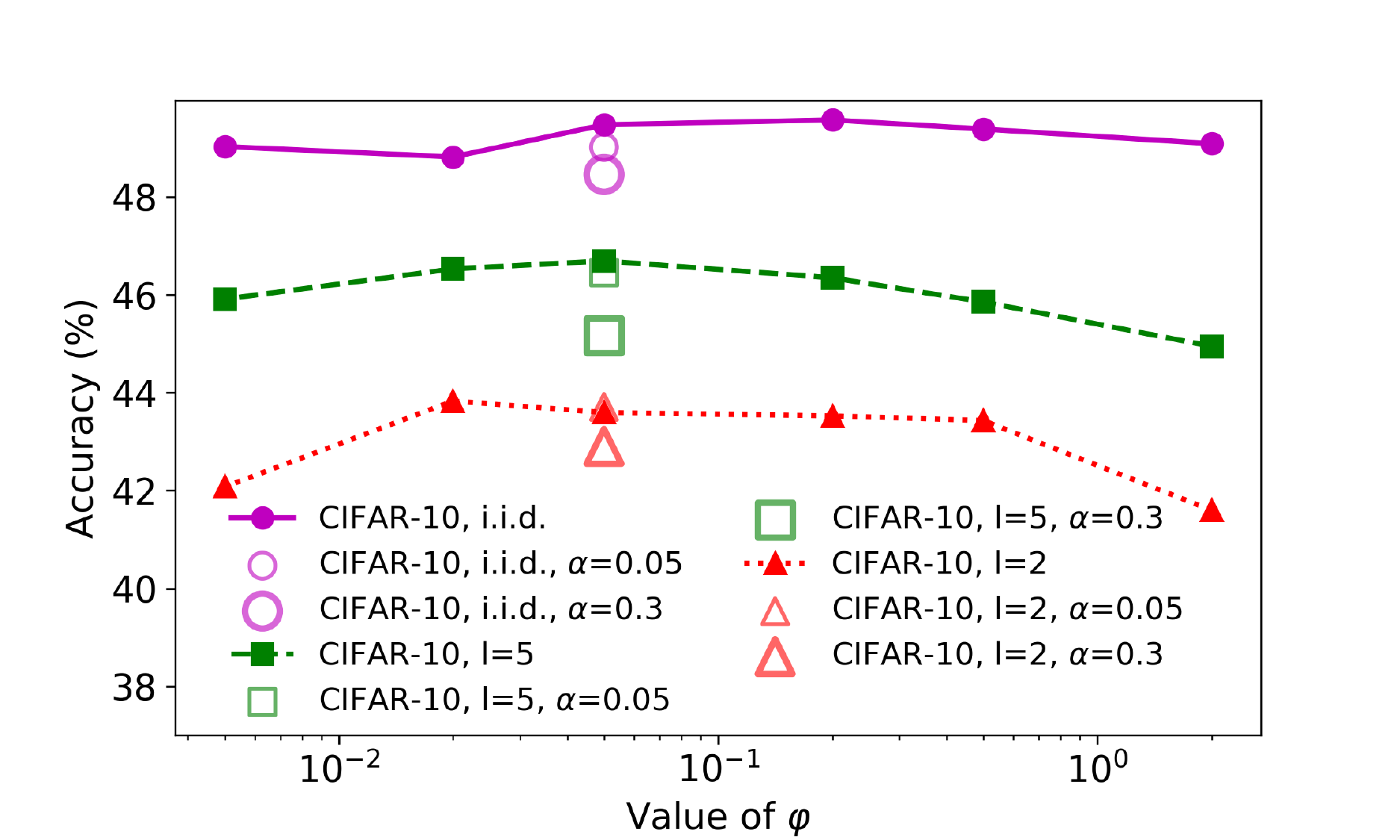}
\label{cifar-phi-acc}}
\hfil
\subfloat[]{
\includegraphics[width=0.42\linewidth]{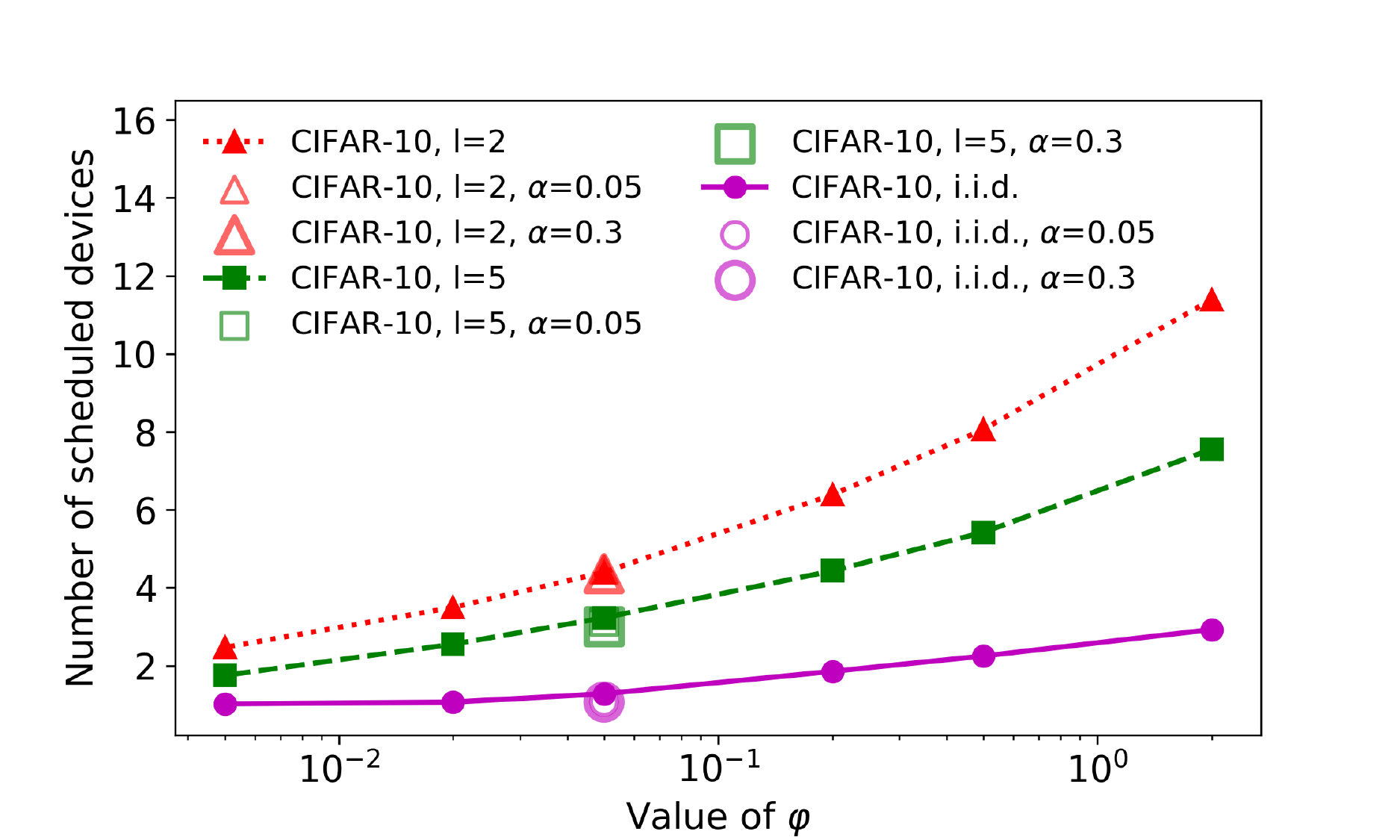}
\label{cifar-phi-scheduleddevices}}
\caption{{Impact of $\varphi$ on the accuracy and number of scheduled devices of proposed FC scheduling.
(a), (c) show the highest achievable accuracy within the training time budget $T$ v.s. the value of $\varphi$ on MNIST and CIFAR-10, respectively.
(b), (d) show the average number of scheduled devices v.s. the value of $\varphi$ on MNIST and CIFAR-10, respectively.
The markers represent the results when the estimation errors of the computation latency $t^\text{cp}_{i,k}$ and the channel state $h_{i, k}$ occur.
$T$ is set to be $60$ seconds for MNIST and $200$ seconds for CIFAR-10. Results are averaged over 5 trails.}}
\label{fig1}
\vspace{-10pt}
\end{figure}

As mentioned in Section III.C, the system parameter $\varphi$ needs to be determined through experiments, thus we study the effects of $\varphi$ first.
{Fig. \ref{fig1}(a), (c) show the highest achievable accuracy within the training time budget v.s. the value of $\varphi$ on MNIST and CIFAR-10, respectively, and Fig. \ref{fig1}(b), (d) show the average number of scheduled devices v.s. the value of $\varphi$ on MNIST and CIFAR-10, respectively.
From Fig. \ref{fig1}(b), (d), we notice that the proposed fast converge scheduling policy (denoted by FC) schedules more devices with larger $\varphi$.}
This finding is due to the following reason.
According to Alg. \ref{greedy}, FC schedules devices by minimizing the objective function of \ref{P5}, in which the term $B(\Pi)$ is not related to $\varphi$ and other terms decrease with $\varphi$.
Therefore, minimizing $B(\Pi)$ is more important with larger $\varphi$, which requires scheduling more devices.
{Furthermore, Fig. \ref{fig1}(a) shows that when $\varphi=0.05$, FC has the best performance in terms of the highest achievable accuracy, which is 92.3\%, 90.6\%, and 89.0\% for i.i.d. and non-i.i.d. data distributions with $l=2$ and $l=1$ on MNIST.
While on CIFAR-10, $\varphi=0.05$ achieves good performance as shown in Fig. \ref{fig1}(c), confirming that FC adapts to different datasets.}
Therefore, we set $\varphi=0.05$ in the following experiments.
{Moreover, it is shown in Fig. \ref{fig1}(a), (c) that the convergence performance is not sensitive to the value of $\varphi$ as long as $0.02\leq \varphi\leq 0.5$, indicating that large step can be taken to reduce the searching cost for $\varphi$ in practice.
The performance of the proposed policy with estimation errors of the computation latency $t^\text{cp}_{i,k}$ and the channel state $h_{i, k}$ is also reported in Fig. \ref{fig1}.
We simulate the estimation error by an Gaussian distribution with $0$ mean and $\alpha t^\text{cp}_{i,k}$ or $\alpha h_{i,k}$ standard deviation for $t^\text{cp}_{i,k}$ or $h_{i, k}$, repectively, where $t^\text{cp}_{i,k}$ and $h_{i, k}$ is the true value.
The average number of scheduled devices with estimation errors is almost the same as that without estimation error as shown in Fig. \ref{fig1}(b), (d), indicating that the performance loss is mostly caused by the round latency.
The main reason for the increasing round latency is that the estimation errors can cause the device scheduling algorithm (Alg. \ref{greedy}) to schedule inappropriate devices.
Nevertheless, Fig. \ref{fig1}(a), (c) show that the proposed policy is robust to estimation errors.
}

\begin{figure}[!t]
\setlength{\abovecaptionskip}{2pt}
\setlength{\belowcaptionskip}{2pt}
  \centering
  \subfloat[]
  {\includegraphics[width=0.30\linewidth]{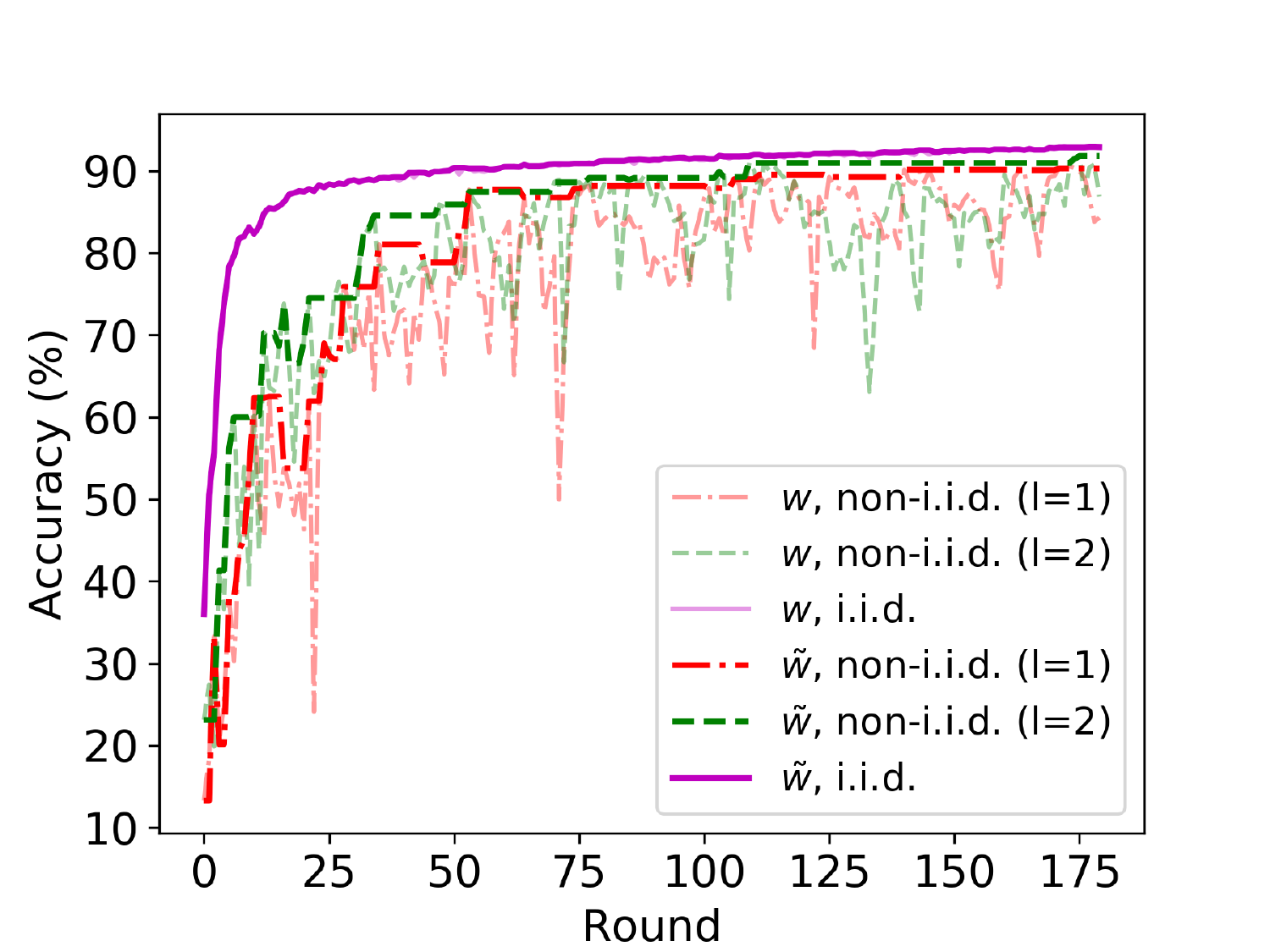}
  \label{round-loss}}
  \hfil
  \subfloat[]
  {\includegraphics[width=0.30\linewidth]{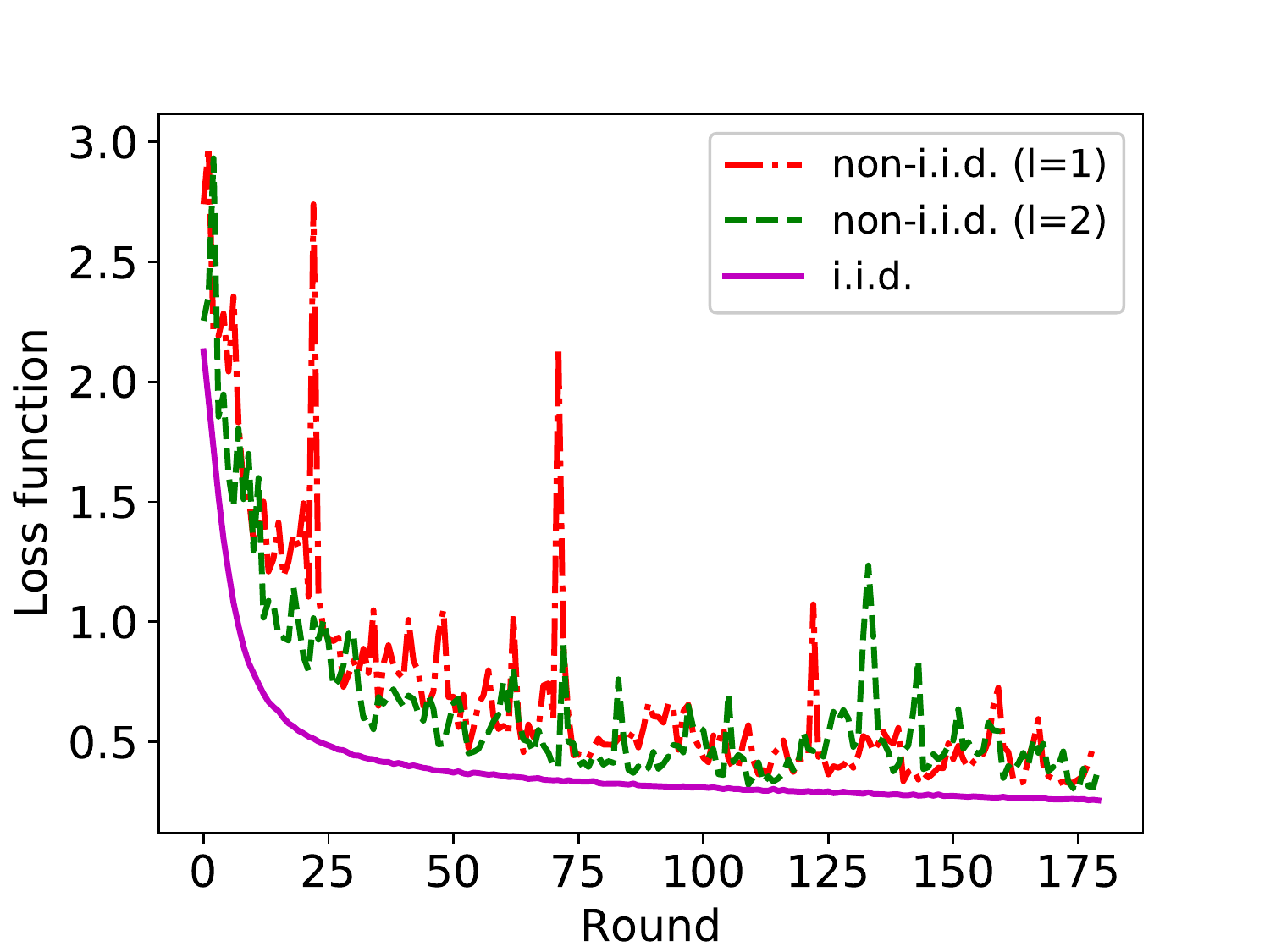}
  \label{round-acc}}
  \hfil
  \subfloat[]
  {\includegraphics[width=0.30\linewidth]{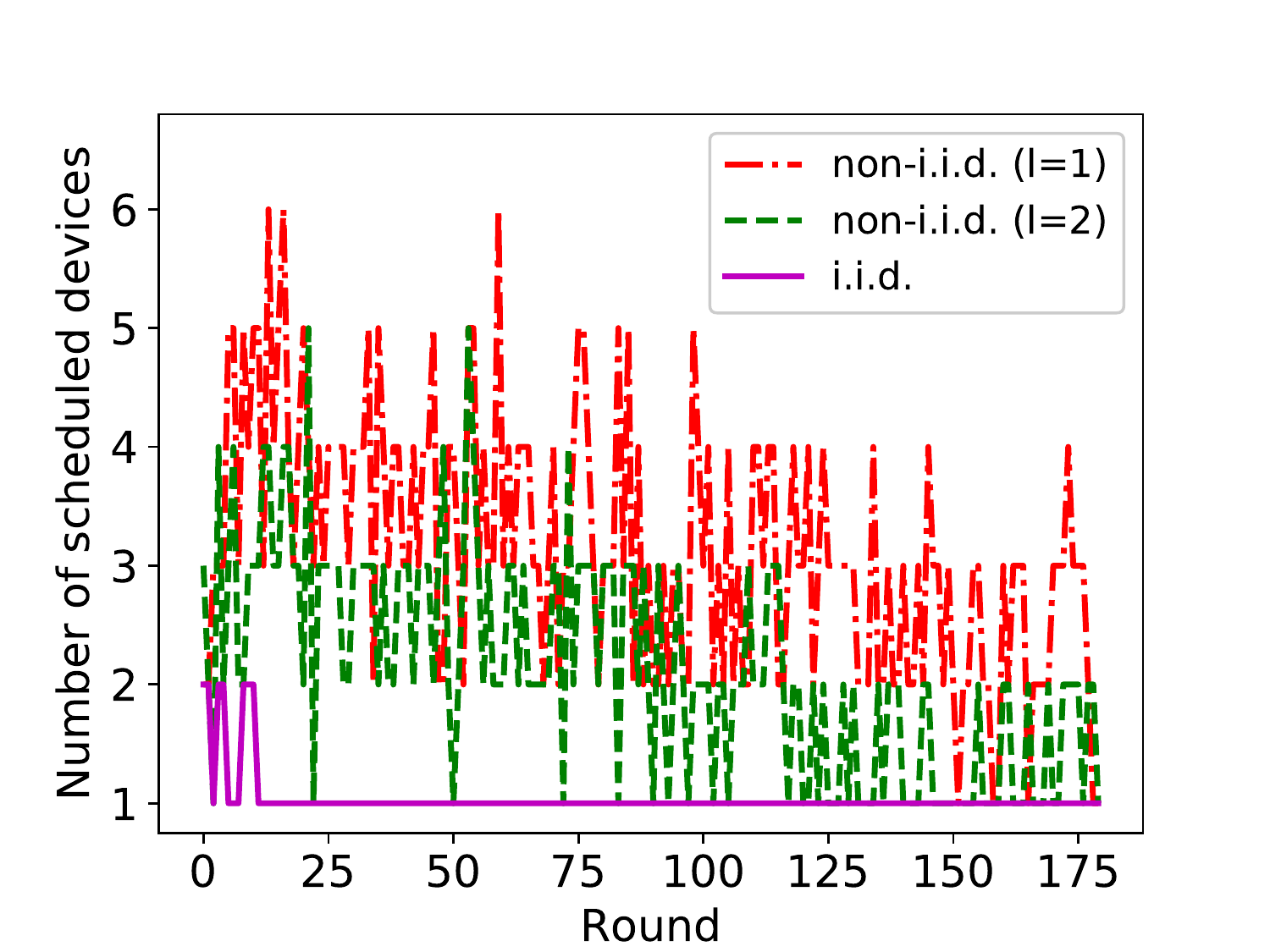}
  \label{round-scheduleddevices}}

  \addvspace{-12pt}
  \subfloat[]
  {\includegraphics[width=0.30\linewidth]{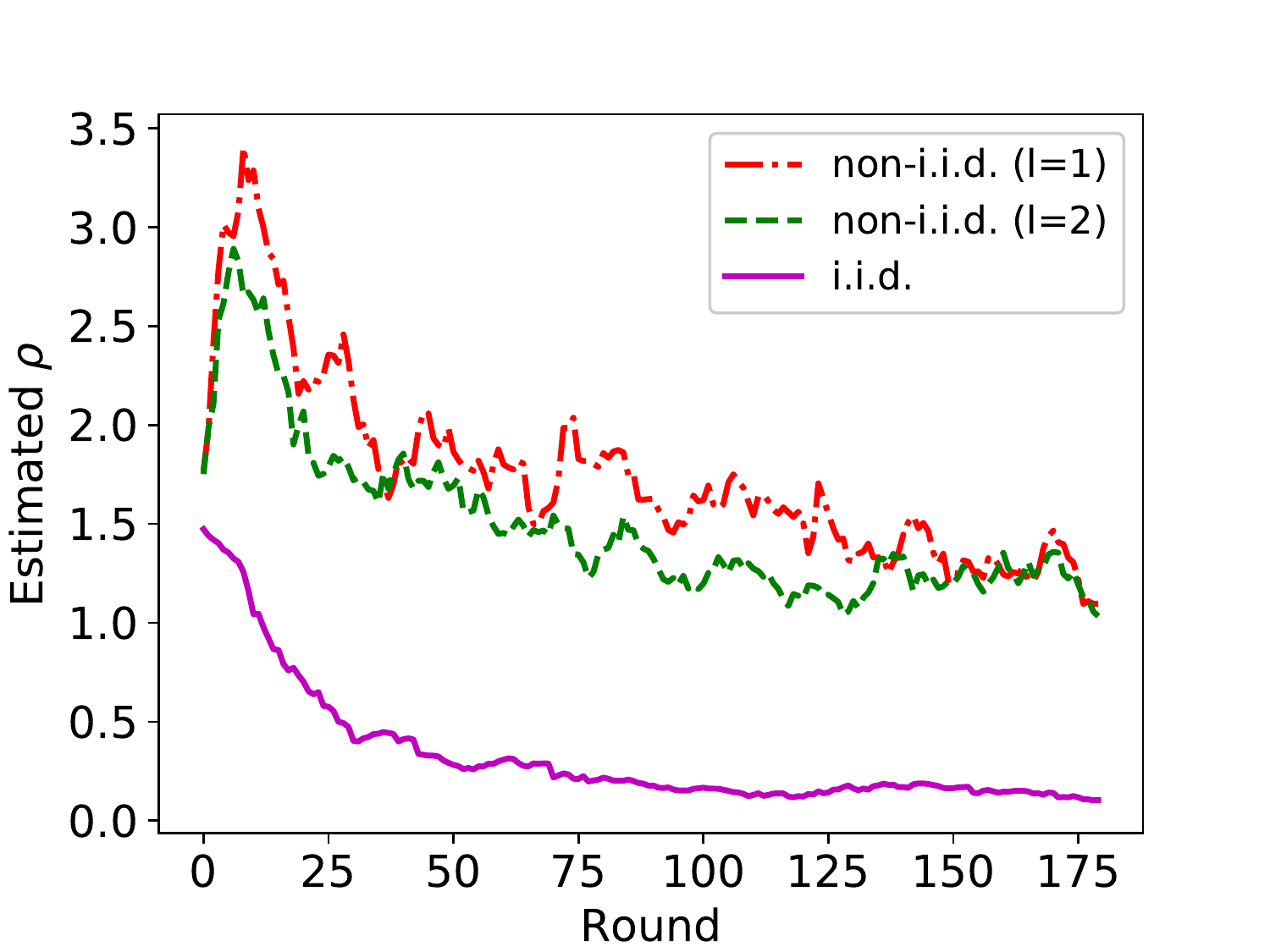}
  \label{round-rho}}
  \hfil
  \subfloat[]
  {\includegraphics[width=0.30\linewidth]{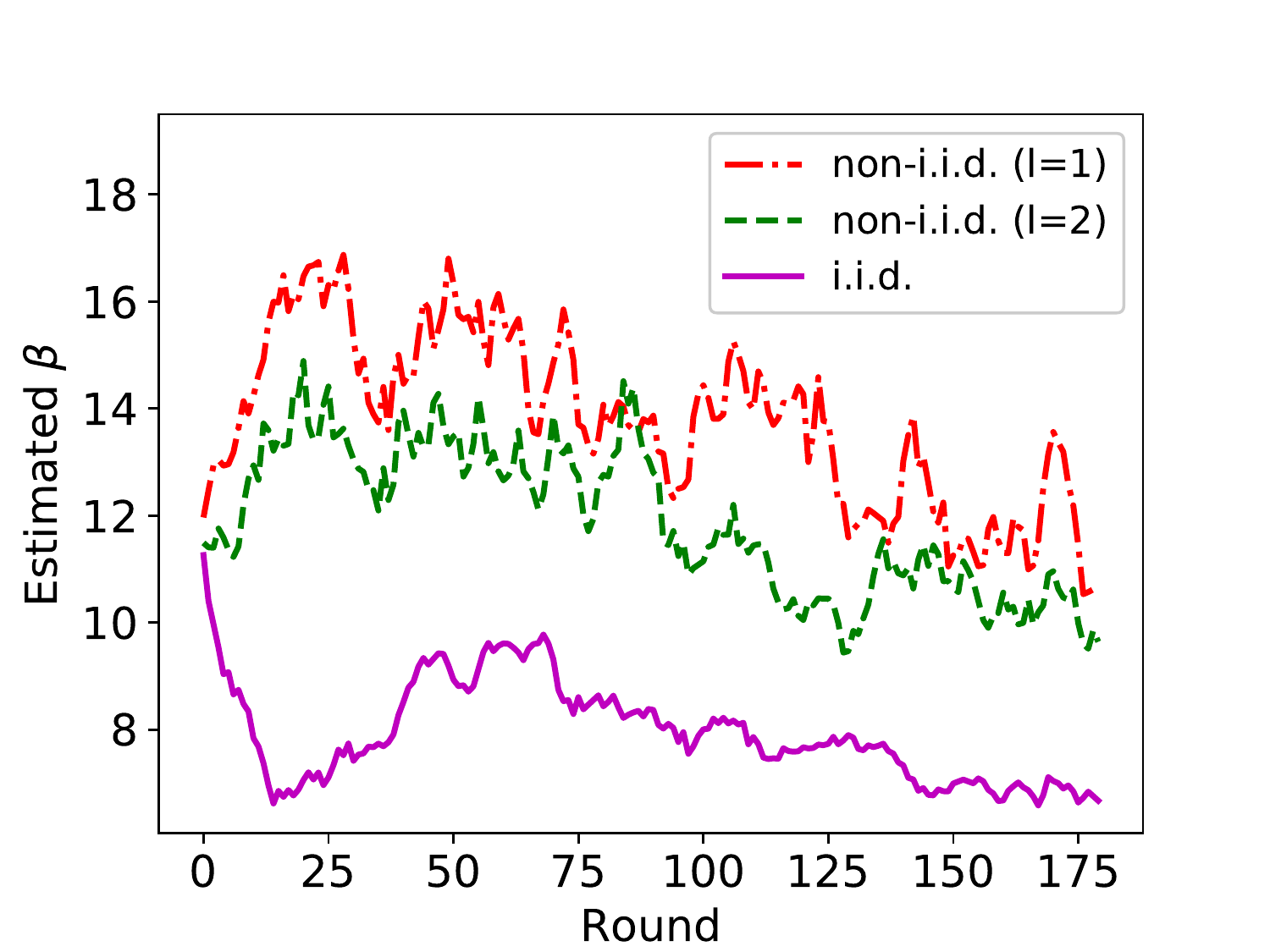}
  \label{5}}
  \hfil
  \subfloat[]
  {\includegraphics[width=0.30\linewidth]{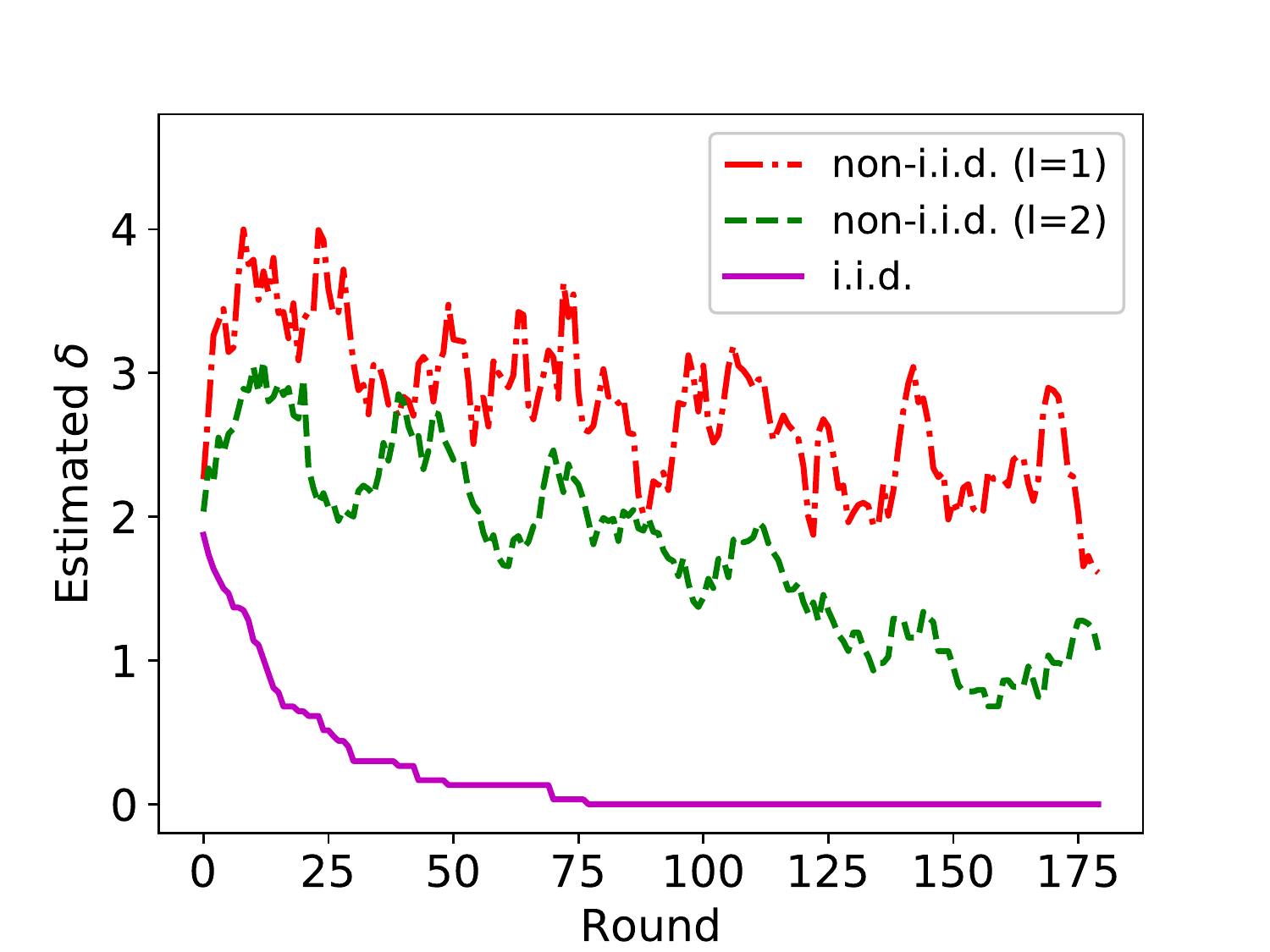}
  \label{6}}
\caption{{Instantaneous results of the proposed fast convergence policy, denoted by FC, on MNIST.} (a), (b), (c), (d), and (e) show the test accuracy, global loss, number of scheduled devices, estimated $\rho$, estimated $\beta$, and estimated $\delta$, respectively.}
\label{fig2}
\vspace{-10pt}
\end{figure}

{Then the instantaneous results of FC on MNIST are shown in Fig. \ref{fig2}.}
Fig. \ref{fig2}(a) shows the model accuracy of $\vec{w}_k^\Pi$ and $\tilde{\vec{w}}$ on the testing dataset v.s. the number of rounds.
Note that the accuracy of $\tilde{\vec{w}}$ is mostly higher than that of $\vec{w}_k^\Pi$ under the same data distribution, which is consistent with the definition of $\tilde{\vec{w}}$.
We also notice that the number of scheduled devices increases with the non-i.i.d. level due to the higher values of the estimated $\rho$, $\beta$, and $\delta$.
For $l=1$ and $l=2$ datasets, since the local datasets of different devices are non-i.i.d., differences between the local updated models are greater than that of the i.i.d. dataset, and thus the value of $\delta$ that characterizes the differences between model updates is higher.
Similar results can be observed for $\beta$ and $\rho$, indicating that the loss function is less smooth and convex for the non-i.i.d. datasets.
Further, since $\rho$, $\beta$, and $\delta$ tend to decrease during the training as shown by Fig. \ref{fig2}(d), (e), and (f), respectively, FC schedules more devices in the beginning of the training process, which helps FL to converge faster \cite{sun2019energy}.

\subsection{Comparison of Different Scheduling Policies}

To show the effectiveness of the convergence analysis, we compare FC with a set of baseline policies that schedule fixed numbers of devices (i.e., remove steps 8-9 in Alg. \ref{greedy}, and stop scheduling new devices until reaching the fixed number).
\begin{figure}[!t]
\setlength{\abovecaptionskip}{2pt}
\setlength{\belowcaptionskip}{2pt}
\centering
\subfloat[]
{\includegraphics[width=0.36\linewidth]{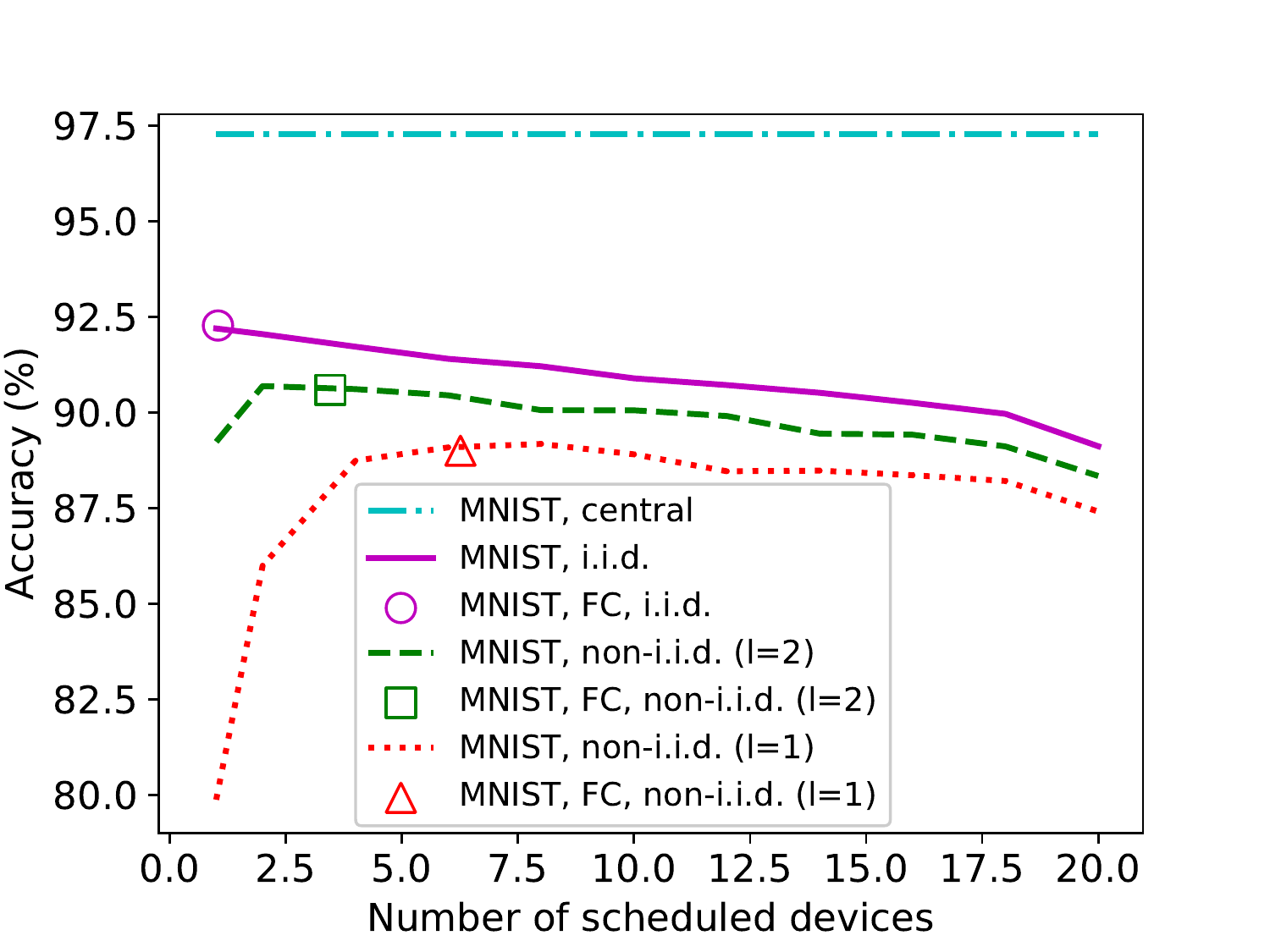}
\label{mnist-frac-acc}}
\hfil
\subfloat[]
{\includegraphics[width=0.36\linewidth]{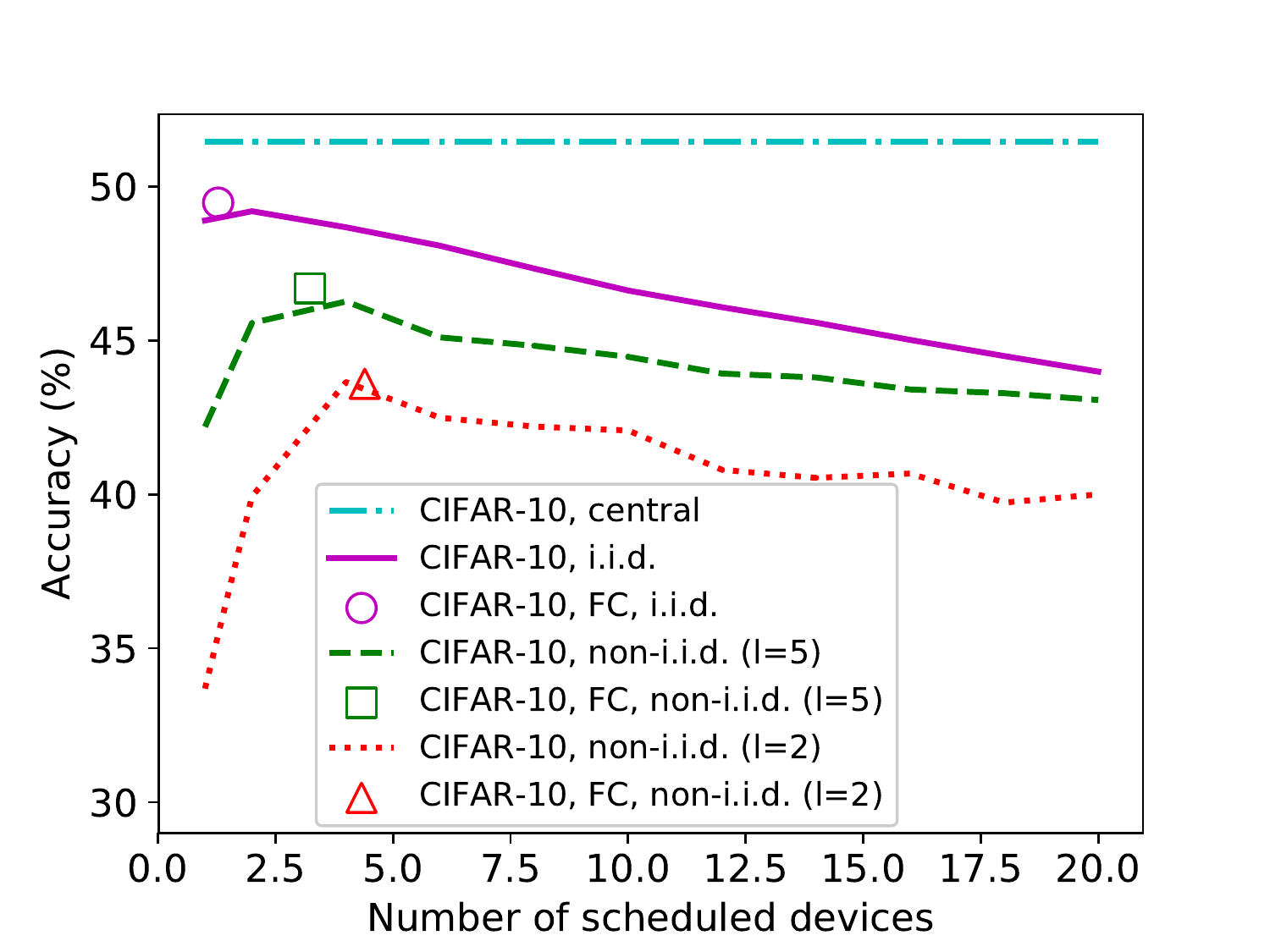}
\label{cifar-frac-acc}}
\caption{{The highest achievable accuracy within the training time budget $T$ v.s. the number of scheduled devices. $T$ is set to be $60$ seconds and $200$ seconds for MNIST and CIFAR-10, respectively. The curves show the results of centralized training and the baseline policies that schedule fixed numbers of devices, and the markers represent FC.} Results are averaged over 5 independent trails.}
\label{frac-acc}
\vspace{-10pt}
\end{figure}
{Fig. \ref{frac-acc} shows the highest achievable accuracy of FC, the baseline policies, and centralized training on MNIST and CIFAR-10, respectively.
The result of centralized training can be treated as the upper bound of performance, where the central trainer is assumed to have 20 times stronger computation capability compared to the devices and all training data have been aggregated to the central trainer.}
We notice that for the baseline policies, scheduling either too few or too many devices degrades the model accuracy for {$l=1$ and $l=2$ on MNIST and all three cases on CIFAR-10}.
The reason is the trade-off between the latency per round and the number of the rounds, that is: scheduling more devices can potentially reduce the number of required rounds to attain a fixed accuracy but with larger latency per round, while scheduling fewer devices can reduce the latency per round but with slower convergence rate w.r.t. the number of rounds.
For FC, since the number of scheduled devices can be optimized, there is only one point for each dataset in the figure which actually corresponds to the average number of scheduled devices.
As shown in Fig. \ref{frac-acc}, FC performs close to the optimal points for all data distributions, because the proposed FC achieves a good trade-off between the latency per round and the number of rounds.
Moreover, the optimal number of scheduled devices increases with the non-i.i.d. level, indicating that a fixed scheduling policy cannot adapt to all different distributions of the datasets.

\begin{figure*}[!t]
\setlength{\abovecaptionskip}{2pt}
\setlength{\belowcaptionskip}{2pt}
  \centering
  \subfloat[]{\includegraphics[width=0.32\linewidth]{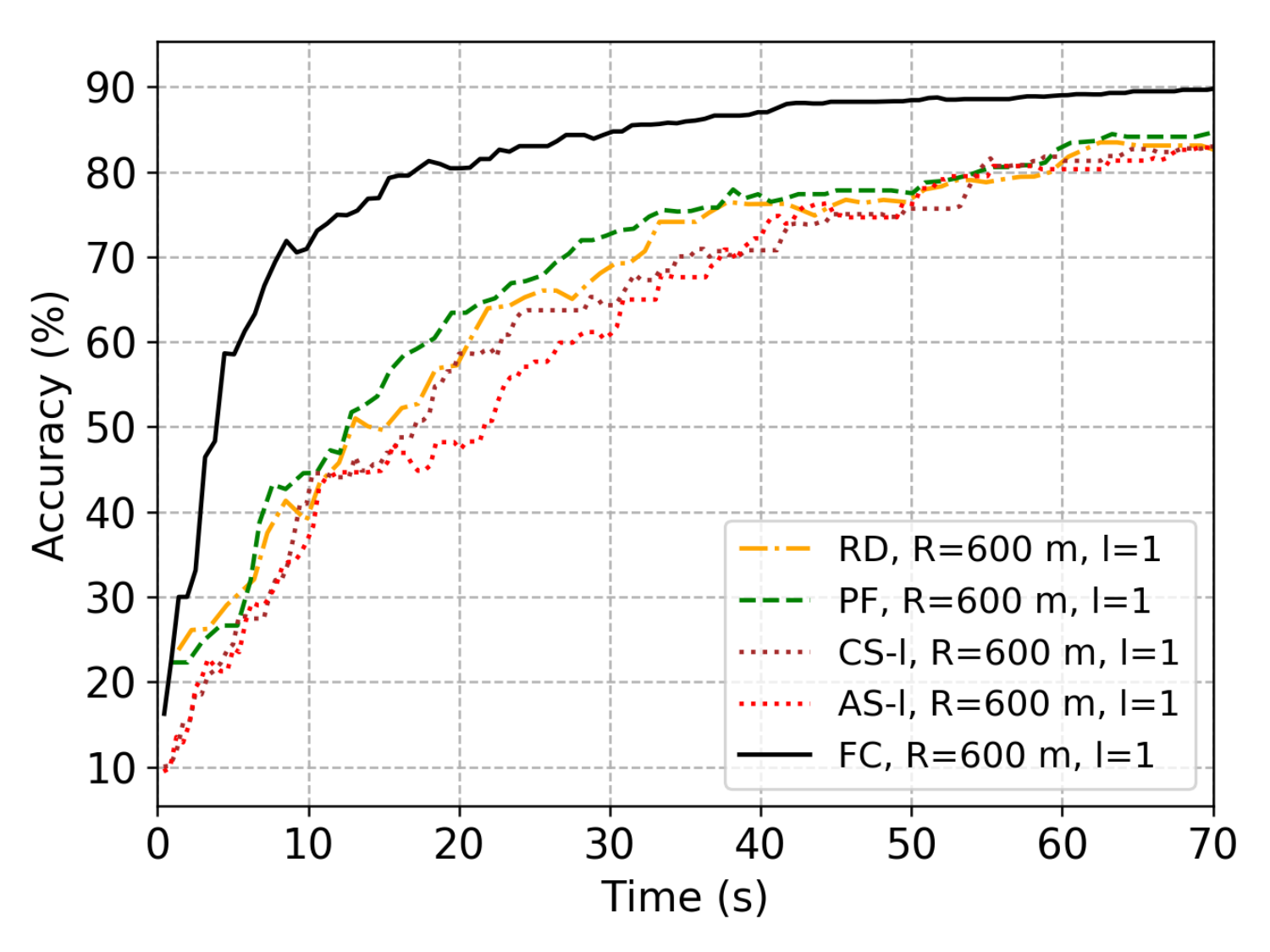}
  \label{time-acc_r200}}
  \hfil
  \subfloat[]{\includegraphics[width=0.32\linewidth]{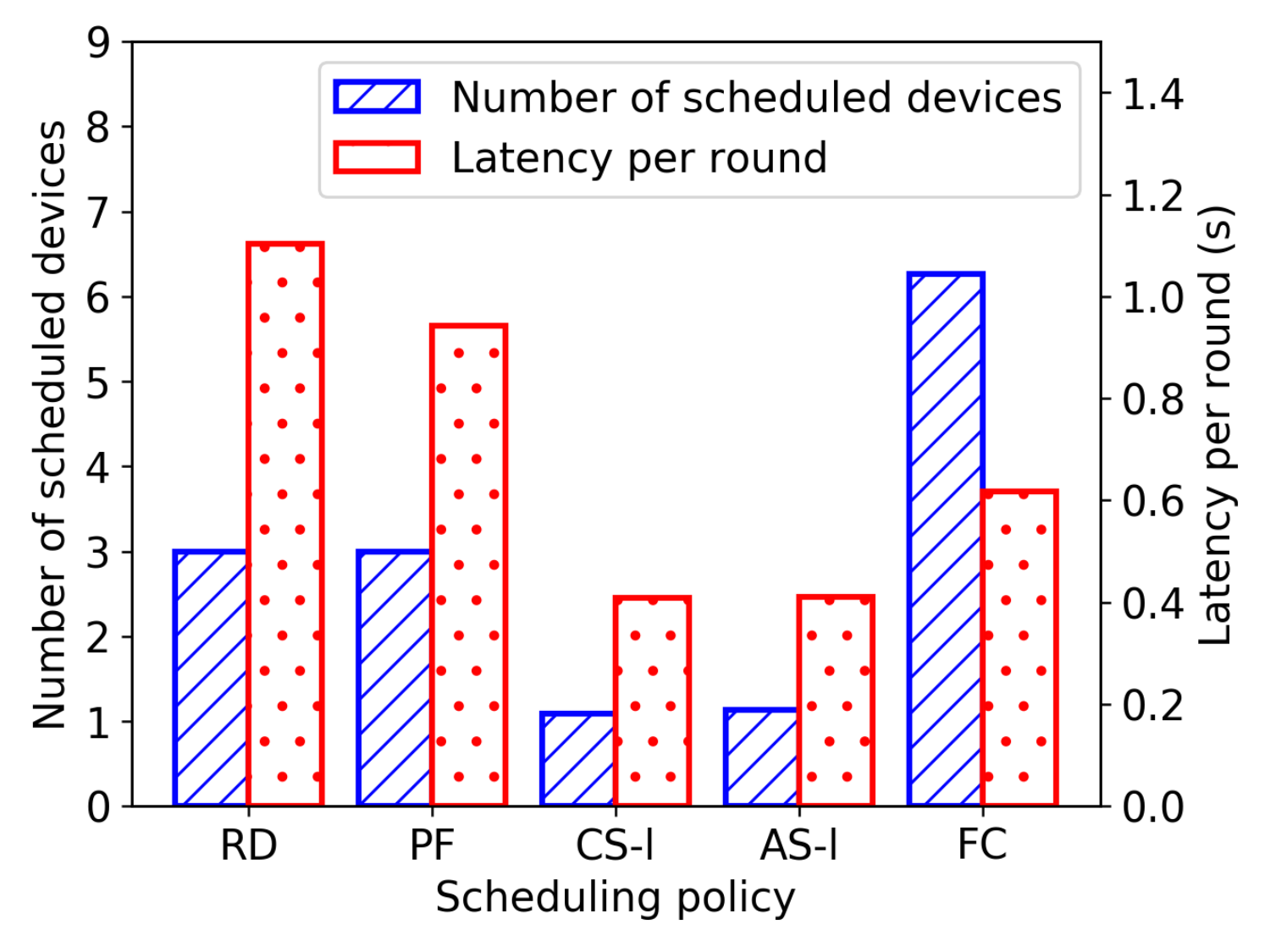}
  \label{avgtime+scheduleddevices_r200}}

  \addvspace{-5pt}
  \subfloat[]{\includegraphics[width=0.32\linewidth]{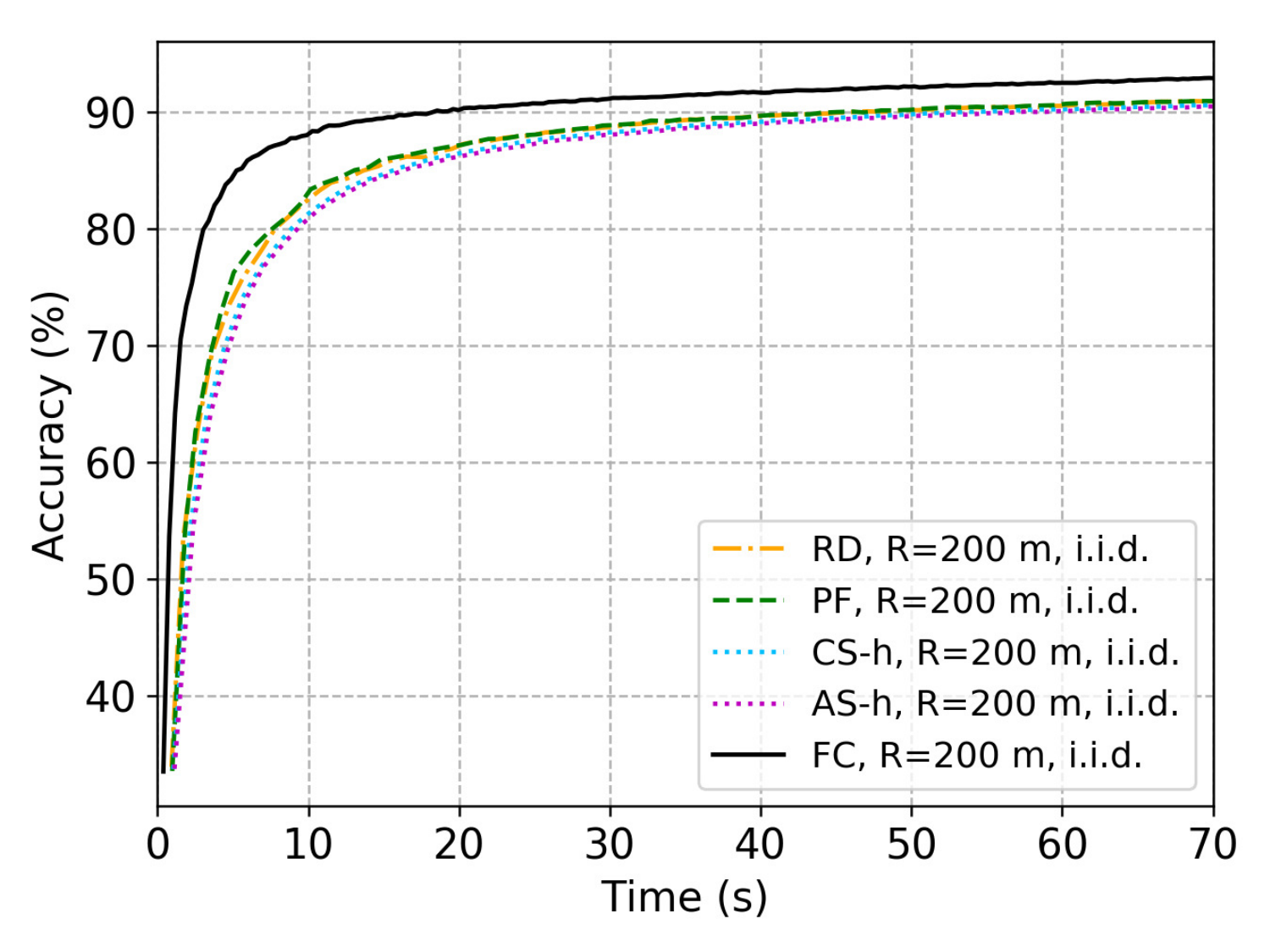}
   \label{time-acc_r600}}
  \hfil
  \subfloat[]{\includegraphics[width=0.32\linewidth]{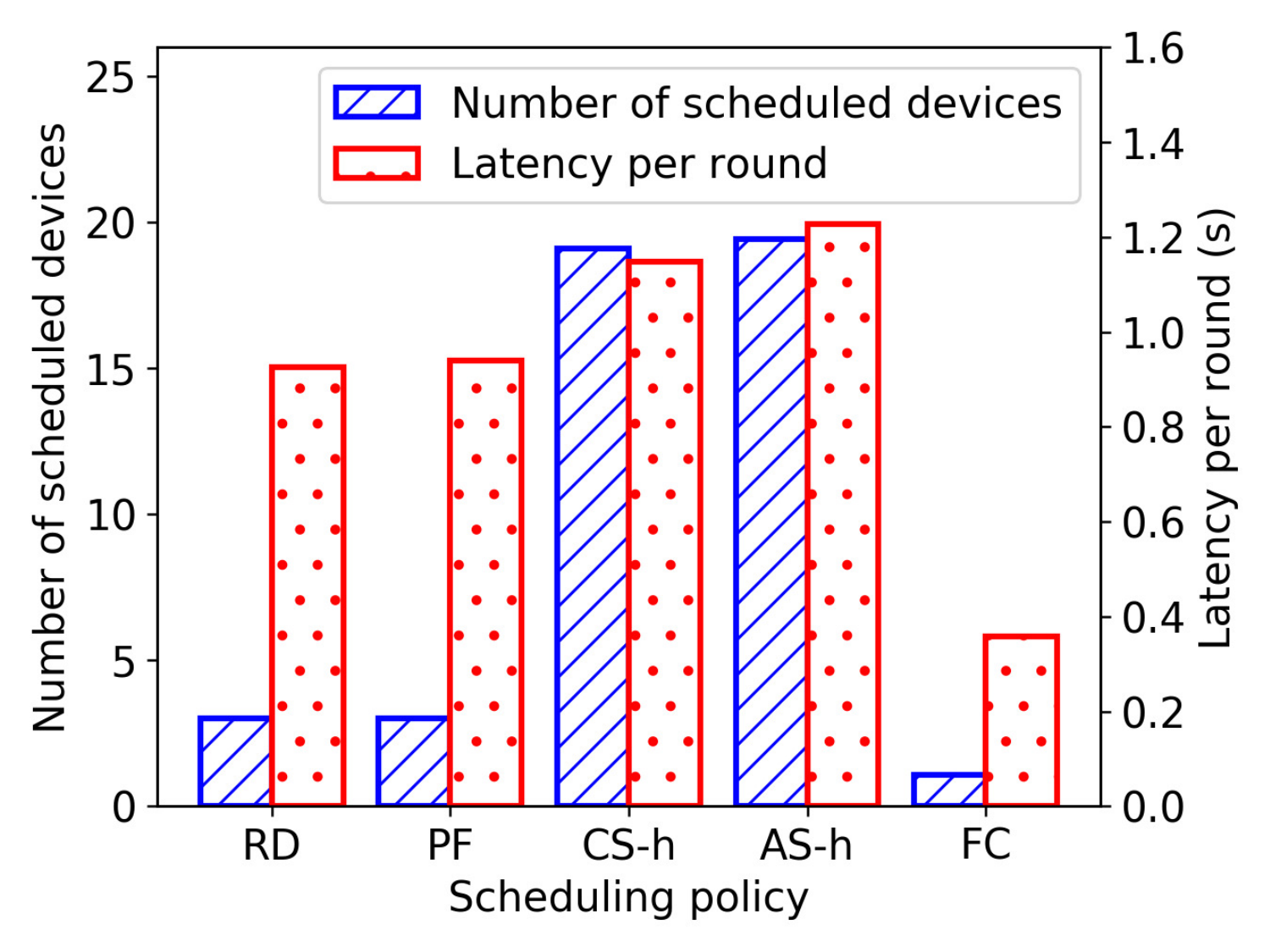}
  \label{avgtime+scheduleddevices_r600}}
  \caption{{The FL convergence performance under different scheduling policies on MNIST.} Results are averaged over 5 independent trails.
  (a) The test accuracy v.s. time with $R=600$ m and $l=1$ dataset.
  (b) The average number of scheduled devices and the corresponding average latency per round w.r.t. different scheduling policies with $R=600$ m and $l=1$ dataset.
  (c) The test accuracy v.s. time with $R=200$ m and i.i.d. dataset.
  (d) The average number of scheduled devices and the corresponding average latency per round w.r.t. different scheduling policies with $R=200$ m and i.i.d. dataset.
  }
  \label{fig3}
\vspace{-10pt}
\end{figure*}

Furthermore, we compare FC with 4 other baseline policies.
The first baseline is the random scheduling policy (denoted by RD) with the empirically optimal number of scheduled devices for {$l=2$ on MNIST}, which is $N_\text{RD}=3$.
The second one is the proportional fair policy (denoted by PF) proposed in \cite{yang2019scheduling} that schedules $N_\text{PF}$ devices with the best instantaneous channel conditions out of all $M$ devices, where we set $N_\text{PF}=N_\text{RD}$ in the experiments.
The third one is the client selection policy proposed in \cite{nishio2019client}, which iteratively schedules the device that consumes the least time in local model updating and uploading, until reaching a preset time threshold $T_{\text{h},\text{CS}}$, and all scheduled devices are allocated equal bandwidth.
Here we use two different thresholds $T_{\text{h},\text{CS}}^\text{low}=0.4$ second and $T_{\text{h},\text{CS}}^\text{high}=1.5$ second, namely CS-l and CS-h, to adapt to various data distributions and cell radius.
The last one is the joint scheduling and resource allocation policy proposed in \cite{chen2019joint} that optimizes the asymptotic convergence performance by scheduling as many devices as possible within a given time threshold $T_{\text{h},\text{AS}}$. In the experiments, we set $T_{\text{h},\text{AS}}^\text{low}=0.4$ second and $T_{\text{h},\text{AS}}^\text{high}=1.5$ second for AS-l and AS-h, respectively.

{The convergence performances w.r.t. time on MNIST under different scheduling policies are reported in Fig. \ref{fig3}(a) and (c) for $R=600$ with $l=1$ dataset and $R=200$ with i.i.d. dataset, respectively.}
Fig. \ref{fig3}(b) and (d) show the corresponding average number of scheduled devices and average latency per round.
For $R=600$ m with $l=1$ dataset, we notice that FC reaches 80\% test accuracy after 17.35 seconds of training, while PF needs 54.71 seconds to attain the same accuracy and other policies are even slower.
Also note that under the given training time budget $T=60$ seconds, the highest achievable accuracy is 89.0\% under FC, which is 9.0\%, 6.4\%, 9.2\%, and 8.1\% higher than RD, PF, CL-l, and AS-l, respectively.
For $R=200$ m with i.i.d. dataset, FC attains 92.6\% test accuracy within the training time budget, exceeds RD, PF, CL-h, and AS-h over 2.1\%, 2.0\%, 2.4\%, and 2.5\%, respectively.
The advantage of FC is twofold: Firstly, FC schedules the devices with better channel conditions and computation capabilities according to Alg. \ref{greedy}, and thus can reduce the per round latency compared to RD and PF.
For example, for $R=600$ m with $l=1$ dataset, FC is able to schedule on average 6.26 devices within 0.62 second per round while PF needs 0.94 second for only 3 devices as shown in Fig. \ref{fig3}(b).
Secondly, FC achieves a better trade-off between the latency per round and the number of required rounds.
Since the time thresholds for CS-l, CS-h, AS-l, and AS-h are fixed, they can hardly adapt to various data distributions and cell radius.
As shown in Fig. \ref{fig3}(a) and (b), CS-l and AS-l schedule too few devices due to the low time threshold, and thus converge slower than FC for $R=600$ m with $l=1$ dataset.
While Fig. \ref{fig3}(c) and (d) show that CS-h and AS-h converge slower than FC for $R=200$ m with i.i.d. dataset because of scheduling too many devices.

\begin{table}
\setlength{\abovecaptionskip}{2pt}
\setlength{\belowcaptionskip}{2pt}
\caption{Highest Achievable Accuracy for All Policies Under Various Data Distributions and Cell Radius.}
\label{tab1}
\begin{center}
\begin{tabular}{c c c c c c c}
\hline
\multirow{2}*{Policy} &  \multicolumn{3}{c}{MNIST} & \multicolumn{3}{c}{CIFAR-10} \\
\cline{2-7}
\multicolumn{1}{c}{} & $R=200$ m & $R=600$ m & $R=1000$ m & $R=200$ m & $R=600$ m & $R=1000$ m\\
\hline
FC (proposed)    & {\bf{89.9}}/91.2/{\bf{92.6}} & {\bf{89.0}}/{\bf{90.6}}/{\bf{92.3}} & 87.0/{\bf{88.6}}/{\bf{92.1}} & {\bf{44.5}}/{\bf{48.2}}/{\bf{50.2}} & {\bf{43.6}}/{\bf{46.7}}/{\bf{49.1}} & {\bf{42.5}}/44.9/{\bf{48.4}}\\
\hline
RD    & 85.4/88.0/90.5 & 80.0/86.0/90.1 & 65.2/78.5/87.8 & 42.0/44.7/47.5 & 39.9/41.3/45.3 & 32.0/33.6/38.0 \\
\hline
PF    & 84.3/88.1/90.6 & 82.6/87.6/90.6 & 81.8/87.6/90.5 & 42.1/45.0/48.0 & 42.1/44.3/47.4 & 41.2/43.5/46.2 \\
\hline
CS-l  & 88.8/91.3/92.5 & 79.8/88.2/92.1 & 75.6/86.3/91.9 & 35.4/43.6/49.6 & 34.9/43.0/48.9 & 32.3/41.4/48.2 \\
\hline
CS-h  & 88.6/89.5/90.2  & 88.7/89.3/90.3 & {\bf{87.1}}/88.4/90.1 & 43.5/46.2/47.9 & 43.3/44.8/47.7 & 41.9/{\bf{45.1}}/47.2\\
\hline
AS-l  & 89.2/{\bf{91.4}}/92.5 & 80.9/88.7/{\bf{92.3}} & 76.1/87.4/92.0 & 37.1/43.9/49.4 & 34.2/43.2/48.8 & 31.3/41.7/47.8 \\
\hline
AS-h  & 88.5/89.3/90.1 & 88.2/89.1/89.9 & 86.9/87.9/89.8 & 42.5/45.6/47.1 & 42.9/44.6/47.0 & 42.0/44.3/46.8 \\
\hline
\end{tabular}
\end{center}
\vspace{-10pt}
\end{table}

{Table \ref{tab1} summarizes the highest achievable accuracy on MNIST and CIFAR-10 under different data distributions and cell radius.
Each item gives the results of $l=1$, $l=2$, and i.i.d. datasets on MNIST and $l=2$, $l=5$, and i.i.d. datasets on CIFAR-10, respectively.}
It is shown that FC adapts to different system settings {and datasets, achieving} the highest accuracy under most settings.
Although CS-l, CS-h, AS-l, and AS-h have similar or even higher (but no more than 0.2\%) accuracy compared to FC under some settings {(e.g., $R=200$ m with $l=2$ on MNIST for CS-l and AS-l, $R=1000$ m with $l=1$ on MNIST for CS-h and AS-h, and $R=1000$ m with $l=5$ on CIFAR-10 for CS-h)}, the accuracy degrades notably under other settings, indicating that CS-l, CS-h, AS-l, and AS-h are not flexible or robust.
Since the optimal number of scheduled devices and the per round latency vary under different system settings, choosing the optimal $T_\text{h}$ for CS and AS accordingly is neither efficient nor practical for wireless FL.

\section{Conclusion}
In this paper, we have studied a joint bandwidth allocation and scheduling problem to optimize the convergence rate of FL w.r.t. time.
We derive a convergence bound of FL to characterize the impact of device scheduling, based on which a joint bandwidth allocation and scheduling policy has been proposed.
The proposed FC policy achieves a desirable trade-off between the latency per round and the number of required rounds, in order to minimize the global loss function of the obtained model parameter under a given training time budget.
The experiments reveal that the optimal number of scheduled devices increases with the non-i.i.d. level of local datasets.
In addition, the proposed FC policy can schedule near-optimal number of devices according to the learned loss function characteristics, gradient characteristics and system dynamics, and outperforms state-of-the-art baseline scheduling policies under different data distributions and cell radius.
{In the future, FL systems with heterogenous device computation capabilities and resource constraints can be considered, where joint optimization of the batch size, the number of local updates, and the device scheduling policy can be studied.}


%

\appendices
\section{Proof of Theorem \ref{thm3}}\label{appendix3}
For $\gamma > 0$, since
\begin{align}
    \frac{\mathrm{d}}{\mathrm{d}\gamma}\left(\gamma B\text{log}_2\left(1+\frac{Ph^2}{\gamma B N_0}\right)\right)
    &= B\text{log}_2\left(1+\frac{Ph^2}{\gamma B N_0}\right) -
    \frac{B P h^2}{(\gamma B N_0 + P h^2)\text{ln}2} \nonumber \\
    & > \frac{B}{\text{ln}2}\left(\frac{\frac{Ph^2}{\gamma B N_0}}{1+\frac{Ph^2}{\gamma B N_0}}\right) - \frac{B P h^2}{(\gamma B N_0 + P h^2)\text{ln}2} \nonumber \\
    & = 0,
\end{align}
where the inequality is because $\text{ln}(1+x) > \frac{x}{1+x}$, for $x > 0$.
Therefore, $\gamma_{i,k}B\text{log}_2\left(1+\frac{P_ih_{i,k}^2}{\gamma_{i,k}BN_0}\right)$ monotonically increases with $\gamma_{i,k}$.
While it is obvious that $\gamma_{i,k}B\text{log}_2\left(1+\frac{P_ih_{i,k}^2}{\gamma_{i,k}BN_0}\right) > 0$ for a non-trivial bandwidth allocation (i.e., $\gamma_{i,k}>0$),
and thus $t_{i,k}^\text{cp}+\frac{S}{\gamma_{i,k}B\text{log}_2\left(1+\frac{P_ih_{i,k}^2}{\gamma_{i,k}BN_0}\right)}$ monotonically decreases with $\gamma_{i,k}$.
{If any device has finished the whole local model update process earlier than other devices,
we can reallocate some bandwidth from that device to other slower devices.
As a result, the round latency, which is determined by the slowest device, can be reduced.
The reallocation of bandwidth can be performed until all devices finish local updating at the same time.
Therefore, the optimal solution of \ref{P2} can be achieved if and only if all bandwidth is allocated and all scheduled devices have the same finishing time.}
As a result, the optimal solution and corresponding objective value is given by the following equations
\begin{equation}
\left\{
\begin{array}{l}
t_{i,k}^\text{cp}+\frac{S}{\gamma^*_{i,k}B\text{log}_2\left(1+\frac{P_ih_{i,k}^2}{\gamma^*_{i,k}BN_0}\right)} = t^*_k(\Pi_k) , \forall i \in \Pi_k, \\
\sum_{i=1}^{M} \gamma^*_{i,k} = 1 ,   \\
0 \leq \gamma^*_{i,k} \leq 1,   \forall i \in[M].
\end{array}
\right.
\label{eq1}
\end{equation}
Solving \eqref{eq1} directly leads to Theorem \ref{thm3}.

\section{Proof of Theorem \ref{thm1}} \label{appendix1}
Based on Assumption \ref{assumption1}, the definition of $F(\vec{w})$, and triangle inequality, we immediately have the following lemma.
\begin{lemma}
If Assumption \ref{assumption1} holds, then $F(\vec{w})$ is convex, $\rho$-Lipschitz, and $\beta$-smooth.
\label{lm}
\end{lemma}
Due to that $F(\vec{w})$ is $\beta$-smooth, we have
\begin{equation}
\setlength\abovedisplayskip{3pt}
\setlength\belowdisplayskip{3pt}
    \mathbb{E}\left\{ F(\vec{w}^\Pi_k) - F(\vec{w}_k)\right\} \leq \frac{\beta}{2}\mathbb{E}\left\Vert \vec{w}^\Pi_k-\vec{w}_k\right\Vert^2.
    \label{step16}
\end{equation}
Substituting $\vec{w}^\Pi_k = \frac{\sum_{i\in\Pi} D_i\vec{w}_{i,k}}{\sum_{i\in\Pi}D_i}$ into the right-hand side of \eqref{step16} yields
\begin{align}
\setlength\abovedisplayskip{3pt}
\setlength\belowdisplayskip{3pt}
    \frac{\beta}{2}\mathbb{E}\left\Vert \vec{w}^\Pi_k-\vec{w}_k\right\Vert^2
    & = \frac{\beta}{2}\mathbb{E}\left\Vert \frac{\sum_{i\in\Pi} D_i\vec{w}_{i,k}}{\sum_{i\in\Pi}D_i}-\vec{w}_k\right\Vert^2 \nonumber\\
    & = \frac{\beta}{2}\mathbb{E}\left\Vert \frac{\sum_{i\in\Pi} D_i(\vec{w}_{i,k}-\vec{w}_k)}{\sum_{i\in\Pi}D_i}\right\Vert^2 \nonumber \\
    & = \frac{\beta}{2}\mathbb{E}\left\Vert \frac{\sum_{i=1}^{M}\mathbb{I}\{i\in\Pi\} D_i(\vec{w}_{i,k}-\vec{w}_k)}{\sum_{i\in\Pi}D_i}\right\Vert^2 \nonumber \\
    & \leq \frac{\beta}{2}\mathbb{E}\left\Vert \frac{\sum_{i=1}^{M}\mathbb{I}\{i\in\Pi\} D_i(\vec{w}_{i,k}-\vec{w}_k)}{|\Pi|D_{\rm min}}\right\Vert^2 \nonumber \\
    & = \frac{\beta}{2} \cdot \frac{\left\Vert\sum_{i=1}^{M}\mathbb{P}\{i\in\Pi\} D_i(\vec{w}_{i,k}-\vec{w}_k)\right\Vert^2}{\left(|\Pi|D_{\rm min}\right)^2},
    \label{step17}
\end{align}
where $\mathbb{I}(\cdot)$ is the indicator function, $\mathbb{P}(\cdot)$ is the probability notation, and $D_{\rm min}\triangleq {\rm min}_{i\in\mathcal{M}}D_i$.
Note that $\Pi$ is a stationary random scheduling policy, and thus $\mathbb{P}\{i\in\Pi\}=\frac{|\Pi|}{M}$ and $\mathbb{P}(i, j\in\Pi, i\neq j)
=\frac{|\Pi|(|\Pi|-1)}{M(M-1)}$.
Therefore, we expand the numerator of the second term of \eqref{step17} as follows:
\begin{align}
\setlength\abovedisplayskip{3pt}
\setlength\belowdisplayskip{3pt}
    &\left\Vert\sum_{i=1}^{M}\mathbb{P}\{i\in\Pi\} D_i(\vec{w}_{i,k}-\vec{w}_k)\right\Vert^2 \nonumber \\
    & = \sum_{i=1}^M \mathbb{P}(i\in\Pi)\Vert D_i(\vec{w}_{i,k}-\vec{w}_k)\Vert^2  + \sum_{i\neq j}\mathbb{P}(i, j\in\Pi)D_i D_j(\vec{w}_{i,k}-\vec{w}_k)^\mathsf{T}(\vec{w}_{j,k}-\vec{w}_k) \nonumber \\
    & = \sum_{i=1}^M \frac{|\Pi|}{M}\Vert D_i(\vec{w}_{i,k}-\vec{w}_k)\Vert^2  + \sum_{i\neq j}\frac{|\Pi|(|\Pi|-1)}{M(M-1)}D_i D_j(\vec{w}_{i,k}-\vec{w}_k)^\mathsf{T}(\vec{w}_{j,k}-\vec{w}_k) \nonumber \\
    & \overset{(a)}{=} \sum_{i=1}^M \frac{|\Pi|}{M}\Vert D_i(\vec{w}_{i,k}-\vec{w}_k)\Vert^2 - \sum_{i=1}^M\frac{|\Pi|(|\Pi|-1)}{M(M-1)}\Vert D_i(\vec{w}_{i,k}-\vec{w}_k)\Vert^2 \nonumber \\
    & = \frac{|\Pi|(M-|\Pi|)}{M(M-1)}\sum_{i=1}^M\Vert D_i(\vec{w}_{i,k}-\vec{w}_k)\Vert^2.
    \label{step18}
\end{align}
The equality of (a) is based on the fact that
\begin{align}
\setlength\abovedisplayskip{3pt}
\setlength\belowdisplayskip{3pt}
    \sum_{i\neq j}D_i D_j(\vec{w}_{i,k}-\vec{w}_k)^T(\vec{w}_{j,k}-\vec{w}_k)
    & = \left\Vert \sum_{i=1}^{M} D_i(\vec{w}_{i,k}-\vec{w}_k)\right\Vert^2 -\sum_{i=1}^M \Vert D_i(\vec{w}_{i,k}-\vec{w}_k)\Vert^2 \nonumber \\
    & = \left\Vert \sum_{i=1}^{M} D_i\vec{w}_{i,k}- \sum_{i=1}^{M} D_i \vec{w}_k\right\Vert^2 -\sum_{i=1}^M \Vert D_i(\vec{w}_{i,k}-\vec{w}_k)\Vert^2 \nonumber \\
    & = -\sum_{i=1}^M \Vert D_i(\vec{w}_{i,k}-\vec{w}_k)\Vert^2 .
    \label{step4}
\end{align}
Furthermore, we bound the term $\sum_{i=1}^M \Vert D_i(\vec{w}_{i,k}-\vec{w}_k)\Vert^2$ as follows:
\begin{align}
\setlength\abovedisplayskip{3pt}
\setlength\belowdisplayskip{3pt}
    \sum_{i=1}^M \Vert D_i(\vec{w}_{i,k}-\vec{w}_k)\Vert^2
    & = \sum_{i=1}^M \left\Vert D_i \left(\vec{w}_{i,k}-\frac{\sum_{j=1}^{M} D_j\vec{w}_{j,k}}{D} \right) \right\Vert^2 \nonumber \\
    & = \sum_{i=1}^M \left\Vert D_i \left(\frac{\sum_{j=1}^{M} D_j(\vec{w}_{i,k}-\vec{w}_{j,k}}{D}\right)\right\Vert^2 \nonumber \\
    & \leq \sum_{i=1}^M \sum_{j=1}^M \left(\frac{D_i^2 D_j^2}{D^2}\Vert \vec{w}_{i,k}-\vec{w}_{j,k} \Vert ^2 \right) \nonumber \\
    & \leq \sum_{i=1}^M \sum_{j=1}^M \left( \frac{D_i^2 D_j^2}{D^2} (\Vert \vec{w}_{i,k}-\vec{v}_k \Vert ^2 + \Vert \vec{w}_{j,k}-\vec{v}_k \Vert ^2) \right) \label{step1}.
\end{align}
Based on Lemma 3 in \cite{wang2019adaptive}, we have $\Vert \vec{w}_{i,k}-\vec{v}_k \Vert \leq g_i(\tau)$ and $\Vert \vec{w}_{j,k}-\vec{v}_k \Vert \leq g_j(\tau)$, where $g_i(x)\triangleq\frac{\delta_i}{\beta}\left((\eta\beta+1)^x-1\right)$.
Substituting into \eqref{step1} yields
\begin{equation}
\setlength\abovedisplayskip{3pt}
\setlength\belowdisplayskip{3pt}
    \sum_{i=1}^M \sum_{j=1}^M \left( \frac{D_i^2 D_j^2}{D^2} (\Vert \vec{w}_{i,k}-\vec{v}_k \Vert ^2 + \Vert \vec{w}_{j,k}-\vec{v}_k \Vert ^2) \right)
    \leq  \frac{\sum_{i=1}^M \sum_{j=1}^M \left( D_i^2 D_j^2 \left(g_i^2(\tau)+g_j^2(\tau)\right)\right)}{D^2}. \label{step3}
\end{equation}
Finally, combining \eqref{step16}, \eqref{step17}, \eqref{step18}, \eqref{step1}, and \eqref{step3} together, we have Theorem \ref{thm1}:
\begin{align}
\setlength\abovedisplayskip{3pt}
\setlength\belowdisplayskip{3pt}
    \mathbb{E}\left\{ F(\vec{w}^\Pi_k) - F(\vec{w}_k)\right\}
    & \leq \frac{\beta(M-|\Pi|)\sum_{i=1}^M \sum_{j=1}^M \left( D_i^2 D_j^2 \left(g_i^2(\tau)+g_j^2(\tau)\right)\right)}{2M(M-1)|\Pi|D_{\rm min}^2 D^2}.
\end{align}

\section{Proof of Theorem \ref{thm2}}\label{appendix2}
First, we have the following lemma.
\begin{lemma}
\label{lm1}
If the following conditions hold:
\begin{enumerate}
    \item $\eta \leq \frac{1}{\beta}$
    \item $\eta\varphi - \frac{\rho h(\tau)+ B(\Pi)}{\tau \epsilon^2}>0$
    \item $F(\vec{v}_k)-F(\vec{w}^*)\geq \epsilon$, $k=1,2,\dots,K$
    \item $F(\vec{w}^\Pi_K)-F(\vec{w}^*)\geq \epsilon$
\end{enumerate}
for some $\epsilon>0$, $\varphi \triangleq \omega\left(1-\frac{\beta\eta}{2}\right)$ and $\omega \triangleq {\rm{min}}_k\frac{1}{\norm{\vec{w}^\Pi_k-\vec{w}^*}}$, then the global loss function of wireless FL can be bounded by
\begin{align}
\setlength\abovedisplayskip{3pt}
\setlength\belowdisplayskip{3pt}
    \mathbb{E}\left\{\frac{1}{F(\vec{w}^{\Pi}_K)-F(\vec{w}^*)}\right\} \geq K \left(\eta\varphi\tau - \frac{\rho h(\tau)+B(\Pi)}{\epsilon^2}\right),
\end{align}
where the expectation is taken over the randomness over $\Pi$.
\end{lemma}
\begin{proof}
First, we define $\theta_k=F(\Vec{v}_k)-F(\vec{w}^*)$. According to (30) in \cite{wang2019adaptive}, we have
\begin{equation}
\begin{aligned}
\setlength\abovedisplayskip{3pt}
\setlength\belowdisplayskip{3pt}
    \frac{1}{\theta_K} - \frac{1}{F(\vec{w}_0)-F(\vec{w}^*)}
    & \geq K\tau\omega\eta \left(1-\frac{\beta\eta}{2}\right) + \sum_{k=1}^{K-1}\left(\frac{1}{F(\vec{w}^{\Pi}_k)-F(\vec{w}^*)}-\frac{1}{\theta_k}\right),
    \label{step5}
\end{aligned}
\end{equation}
where $\omega\triangleq {\rm{min}}_k\frac{1}{\norm{\vec{w}^\Pi_k-\vec{w}^*}}$.
Each term in the summation in the right-hand side of \eqref{step5} can be further expressed as
\begin{align}
\setlength\abovedisplayskip{3pt}
\setlength\belowdisplayskip{3pt}
    \frac{1}{F(\vec{w}^{\Pi}_k)-F(\vec{w}^*)}-\frac{1}{\theta_k}
    & = \frac{\theta_k-\left(F(\vec{w}^{\Pi}_k)-F(\vec{w}^*)\right)}{\left(F(\vec{w}^{\Pi}_k)-F(\vec{w}^*)\right)\theta_k}  = \frac{F(\Vec{v}_k)-F(\vec{w}^{\Pi}_k)}{\left(F(\vec{w}^{\Pi}_k)-F(\vec{w}^*)\right)\theta_k} \nonumber \\
    & = \frac{\left(F(\Vec{v}_k)-F(\vec{w}_k)\right)-\left(F(\vec{w}^{\Pi}_k)-F(\vec{w}_k)\right)}{\left(F(\vec{w}^{\Pi}_k)-F(\vec{w}^*)\right)\theta_k} \nonumber \\
    & \geq \frac{-\rho h(\tau)-\left(F(\vec{w}^{\Pi}_k)-F(\vec{w}_k)\right)}{\left(F(\vec{w}^{\Pi}_k)-F(\vec{w}^*)\right)\theta_k},
    \label{step6}
\end{align}
where the last inequality is due to Theorem 1 in \cite{wang2019adaptive}.
Assume that $\theta_k = F(\Vec{v}_k)-F(\vec{w}^*)\geq \epsilon$ for all $k$.
By summing up (26) in \cite{wang2019adaptive} for all $\tau$ steps of centralized gradient descent of $\Vec{v}_k$, we have $F(\Vec{v}_k) \leq F(\vec{w}^\Pi_{k-1})$.
Therefore, $F(\vec{w}^{\Pi}_k)-F(\vec{w}^*) \geq F(\Vec{v}_{k+1})-F(\vec{w}^*) \geq \epsilon $, and thus $\left(F(\vec{w}^{\Pi}_k)-F(\vec{w}^*)\right)\theta_k \geq \epsilon^2$, consequently
\begin{align}
\setlength\abovedisplayskip{3pt}
\setlength\belowdisplayskip{3pt}
    \frac{-1}{\left(F(\vec{w}^{\Pi}_k)-F(\vec{w}^*)\right)\theta_k} &\geq -\frac{1}{\epsilon^2}.
    \label{step7}
\end{align}
Substituting \eqref{step7} into \eqref{step6} yields
\begin{align}
\setlength\abovedisplayskip{3pt}
\setlength\belowdisplayskip{3pt}
    \frac{1}{F(\vec{w}^{\Pi}_k)-F(\vec{w}^*)}-\frac{1}{\theta_k}
    &\geq \frac{-\rho h(\tau)-\left(F(\vec{w}^{\Pi}_k)-F(\vec{w}_k)\right)}{\epsilon^2} \nonumber.
\end{align}
Then take expectation over the randomness of $\Pi$ and apply Theorem \ref{thm1}, we have
\begin{equation}
\setlength\abovedisplayskip{3pt}
\setlength\belowdisplayskip{3pt}
  \mathbb{E}\left\{\frac{1}{F(\vec{w}^{\Pi}_k)-F(\vec{w}^*)}-\frac{1}{\theta_k} \right\}
  \geq \frac{-\rho h(\tau)-B(\Pi)}{\epsilon^2}.
  \label{step8}
\end{equation}
Then substitute \eqref{step8} into \eqref{step5} and take expectation over the randomness of $\Pi$, we have
\begin{align}
\setlength\abovedisplayskip{3pt}
\setlength\belowdisplayskip{3pt}
    \mathbb{E} \left\{\frac{1}{\theta_K} - \frac{1}{F(\vec{w}_0)-F(\vec{w}^*)}\right\}
    & \geq K\tau\omega\eta \left(1-\frac{\beta\eta}{2}\right) + \sum_{k=1}^{K-1}\mathbb{E}\left\{\frac{1}{F(\vec{w}^{\Pi}_k)-F(\vec{w}^*)}-\frac{1}{\theta_k} \right\} \nonumber \\
    & \geq K\tau\omega\eta \left(1-\frac{\beta\eta}{2}\right) + (K-1) \frac{-\rho h(\tau)-B(\Pi)}{\epsilon^2} .
    \label{step11}
\end{align}
Also assume that $F(\vec{w}^\Pi_K)-F(\vec{w}^*)\geq \epsilon$. Similar to the argument as for obtaining \eqref{step7}, we have
\begin{equation}
\setlength\abovedisplayskip{3pt}
\setlength\belowdisplayskip{3pt}
    \frac{-1}{\left(F(\vec{w}^{\Pi}_K)-F(\vec{w}^*)\right)\theta_K} \geq -\frac{1}{\epsilon^2}.
    \label{step9}
\end{equation}
Subsequently, we have
\begin{align}
\setlength\abovedisplayskip{3pt}
\setlength\belowdisplayskip{3pt}
    \frac{1}{F(\vec{w}^{\Pi}_K)-F(\vec{w}^*)}-\frac{1}{\theta_K}
    & = \frac{\theta_K-\left(F(\vec{w}^{\Pi}_K)-F(\vec{w}^*)\right)}{\left(F(\vec{w}^{\Pi}_K)-F(\vec{w}^*)\right)\theta_K}  = \frac{F(\Vec{v}_K)-F(\vec{w}^{\Pi}_K)}{\left(F(\vec{w}^{\Pi}_K)-F(\vec{w}^*)\right)\theta_K} \nonumber \\
    & = \frac{\left(F(\Vec{v}_K)-F(\vec{w}_K)\right)-\left(F(\vec{w}^{\Pi}_K)-F(\vec{w}_K)\right)}{\left(F(\vec{w}^{\Pi}_K)-F(\vec{w}^*)\right)\theta_K} \nonumber \\
    & \overset{(a)}{\geq} \frac{-\rho h(\tau)-\left(F(\vec{w}^{\Pi}_k)-F(\vec{w}_k)\right)}{\left(F(\vec{w}^{\Pi}_k)-F(\vec{w}^*)\right)\theta_k} \nonumber \\
    & \overset{(b)}{\geq} \frac{-\rho h(\tau)-\left(F(\vec{w}^{\Pi}_k)-F(\vec{w}_k)\right)}{\epsilon^2} \label{step_pre_10}.
\end{align}
The inequality of (a) is due to Theorem 1 in \cite{wang2019adaptive}, and the inequality of (b) is due to \eqref{step9}.
Substituting Theorem \ref{thm1} into \eqref{step_pre_10} and taking expectation over the randomness of $\Pi$ yield
\begin{equation}
\setlength\abovedisplayskip{3pt}
\setlength\belowdisplayskip{3pt}
    \mathbb{E}\left\{\frac{1}{F(\vec{w}^{\Pi}_K)-F(\vec{w}^*)}-\frac{1}{\theta_K}\right\}
    \geq \frac{-\rho h(\tau)-B(\Pi)}{\epsilon^2}. \label{step10}
\end{equation}
Then sum up \eqref{step10} and \eqref{step11}, we have
\begin{align}
\setlength\abovedisplayskip{3pt}
\setlength\belowdisplayskip{3pt}
    \mathbb{E}\left\{\frac{1}{F(\vec{w}^{\Pi}_K)-F(\vec{w}^*)}- \frac{1}{F(\vec{w}_0)-F(\vec{w}^*)}\right\}
    \geq K \left(\tau\omega\eta \left(1-\frac{\beta\eta}{2}\right) - \frac{\rho h(\tau)+B(\Pi)}{\epsilon^2}\right).
\end{align}
Since $F(\vec{w}_0)-F(\vec{w}^*)>0$ due to the definition of $\vec{w}^*$, we have
\begin{align}
\setlength\abovedisplayskip{3pt}
\setlength\belowdisplayskip{3pt}
    \mathbb{E}\left\{\frac{1}{F(\vec{w}^{\Pi}_K)-F(\vec{w}^*)}\right\}
    & \geq \mathbb{E}\left\{\frac{1}{F(\vec{w}^{\Pi}_K)-F(\vec{w}^*)}- \frac{1}{F(\vec{w}_0)-F(\vec{w}^*)}\right\} \nonumber \\
    & \geq K \left(\tau\omega\eta \left(1-\frac{\beta\eta}{2}\right) - \frac{\rho h(\tau)+B(\Pi)}{\epsilon^2}\right) \nonumber \\
    & = K \left(\eta\varphi\tau - \frac{\rho h(\tau)+B(\Pi)}{\epsilon^2}\right) \label{step12}.
    \end{align}
\end{proof}
Based on Lemma \ref{lm1}, we can now proof Theorem \ref{thm2}.
Since the condition $\eta \leq \frac{1}{\beta}$ in Theorem \ref{thm2}, the first condition in Lemma \ref{lm1} is always satisfied.

When $\rho h(\tau) + B(\Pi) = 0 $, $\epsilon$ can be chosen arbitrarily small to satisfy conditions 2-4 in Lemma \ref{lm1}. Because the right hand side of \eqref{eqthm2} and \eqref{step12} is equal when $\rho h(\tau) + B(\Pi) = 0 $, Theorem \ref{thm2} is directly from Lemma \ref{lm1}.

When $\rho h(\tau) + B(\Pi) > 0 $, we consider the right hand side of \eqref{step12} and let
\begin{equation}
\setlength\abovedisplayskip{3pt}
\setlength\belowdisplayskip{3pt}
    \frac{1}{\epsilon_0} = K \left(\eta\varphi\tau - \frac{\rho h(\tau)+B(\Pi)}{\epsilon_0^2}\right).
    \label{step13}
\end{equation}
Solving for $\epsilon_0$ and ignoring the negative solution, we have
\begin{equation}
\setlength\abovedisplayskip{3pt}
\setlength\belowdisplayskip{3pt}
    \epsilon_0 =  \frac{1+\sqrt{1+4\eta\varphi K^2\tau\left(\rho h(\tau)+B(\Pi)\right)}}{2\eta\varphi K \tau}.
    \label{step14}
\end{equation}
Since $\epsilon_0>0$, $\eta\varphi\tau - \frac{\rho h(\tau)+B(\Pi)}{\epsilon_0^2}>0$ based on \eqref{step13}.
We note that $\eta\varphi\tau - \frac{\rho h(\tau)+B(\Pi)}{\epsilon^2}$ increases with $\epsilon$ when $\rho h(\tau)+B(\Pi)>0$, hence condition 2 in Lemma \ref{lm1} is satisfied for any $\epsilon>\epsilon_0$.

Suppose that there exists $\epsilon>\epsilon_0$ satisfying conditions 3 and 4 in Lemma \ref{lm1}. Then all the conditions in Lemma \ref{lm1} are satisfied, and we have
\begin{align}
\setlength\abovedisplayskip{3pt}
\setlength\belowdisplayskip{3pt}
    \mathbb{E}\left\{\frac{1}{F(\vec{w}^{\Pi}_K)-F(\vec{w}^*)}\right\} &\geq K \left(\eta\varphi\tau - \frac{\rho h(\tau)+B(\Pi)}{\epsilon^2}\right) \geq K \left(\eta\varphi\tau - \frac{\rho h(\tau)+B(\Pi)}{\epsilon_0^2}\right) = \frac{1}{\epsilon_0}.
\end{align}
It contradicts with condition 4 in Lemma \ref{lm1}, and thus there does not exist $\epsilon>\epsilon_0$ that satisfy both conditions 3 and 4 in Lemma \ref{lm1}. Therefore, either $\exists k$ that satisfies $F(\vec{v}_k)-F(\vec{w}^*)\leq \epsilon_0$ or $F(\vec{v}_k)-F(\vec{w}^*)\leq \epsilon_0$.
Then we have
\begin{equation}
\setlength\abovedisplayskip{3pt}
\setlength\belowdisplayskip{3pt}
    {\rm min}\left\{ \underset{k=1,2,\dots,K}{\rm min}F(\vec{v}_k);F(\vec{w}^\Pi_K)\right\}-F(\vec{w}^*)\leq \epsilon_0.
    \label{step15}
\end{equation}
Based on Theorem \ref{thm1} and Theorem 1 in \cite{wang2019adaptive}, we have $\mathbb{E}\{F(\vec{w}^\Pi_k)\}\leq F(\vec{v}_k)+\rho h(\tau) + B(\Pi)$ for any $k$. Substituting into \eqref{step15} yields
\begin{equation}
\setlength\abovedisplayskip{3pt}
\setlength\belowdisplayskip{3pt}
    \underset{k=1,2,\dots,K}{\rm min}\mathbb{E}\left\{F(\vec{w}^\Pi_k) - F(\vec{w}^*)\right\}
    \leq \epsilon_0 + \rho h(\tau) + B(\Pi).
\end{equation}
Then based on the Jensen's inequality and the convexity of $f(x)=\frac{1}{x}$, we have
\begin{align}
\setlength\abovedisplayskip{3pt}
\setlength\belowdisplayskip{3pt}
    \mathbb{E}\left\{ \frac{1}{\underset{k=1,2,\dots,K}{\rm min}F(\vec{w}^\Pi_k) - F(\vec{w}^*)}\right\}
    & \geq \frac{1}{\mathbb{E}\left\{\underset{k=1,2,\dots,K}{\rm min}F(\vec{w}^\Pi_k) - F(\vec{w}^*)\right\}} \nonumber \\
    & \geq \frac{1}{\underset{k=1,2,\dots,K}{\rm min}\mathbb{E}\left\{F(\vec{w}^\Pi_k) - F(\vec{w}^*)\right\}} \nonumber \\
    & \geq \frac{1}{\epsilon_0 + \rho h(\tau) + B(\Pi)}.
    \label{step19}
\end{align}
Combine \eqref{step19} with the definition of $F(\tilde{\vec{w}})$ (i.e., \eqref{def-tildew}) and \eqref{step14}, we have Theorem \ref{thm2} proved.

\ifCLASSOPTIONcaptionsoff
  \newpage
\fi



%
\bibliographystyle{IEEEtran}
\bibliography{reference}

%








\end{document}